\documentclass[11pt]{article}

\usepackage{thumbpdf, amssymb, amsmath, amsthm, microtype, array,
graphicx, verbatim, listings, color, fancybox}
\usepackage[pdftex]{hyperref}
\usepackage{mathtools}
\usepackage{booktabs}
\usepackage{setspace}

\usepackage[normalem]{ulem}

\usepackage[linesnumbered,ruled]{algorithm2e}
\usepackage{subcaption}

\usepackage{caption}
\usepackage{subcaption}
\usepackage{thmtools}
\usepackage{bbm}
\usepackage{pgf,tikz}
\usetikzlibrary{arrows}
\usepackage[margin=1.2in]{geometry}
\usepackage{multicol,multirow}

\usepackage{mathrsfs}

\usepackage{thm-restate}
\usepackage{cleveref}

\newcommand{\p}{\rho}

\newcommand{\pel}{p_{\ell}}

\newcommand{\Zplus}{\mathbb{Z}_+}
\newcommand{\rec}{{\gamma}}

\newcommand{\C}{{\mathcal{C}}}

\newcommand{\EC}{\mathrm{2EC}}
\newcommand{\HK}{\mathrm{LP}}

\newcommand{\cov}{\mathrm{cov}}
\newcommand{\point}{\mathrm{end}}

\newcommand{\gap}{\alpha^{\HK}_{\EC}}
\newcommand{\Lu}{\mathcal{L}_\ell} 
\newcommand{\Rv}{\mathcal{R}_\ell} 

\newcommand{\Lup}{\mathcal{L}_{\ell'}} 
\newcommand{\Rvp}{\mathcal{R}_{\ell'}}

\declaretheoremstyle[style=claim,qed=$\Diamond$]{claim}
\declaretheoremstyle[style=plain,qed=$\square$]{theorem}

\theoremstyle{plain}
\newtheorem{theorem}{Theorem}
\newtheorem{corollary}[theorem]{Corollary}
\newtheorem{lemma}[theorem]{Lemma}
\newtheorem{conjecture}{Conjecture}

\newtheorem{observation}[theorem]{Observation}

\newtheorem{definition}[theorem]{Definition}
\newtheorem{claim}{Claim}

\numberwithin{equation}{section}
\numberwithin{theorem}{section}
\numberwithin{lemma}{section}
\numberwithin{corollary}{section}
\numberwithin{definition}{section}
\numberwithin{observation}{section}

\DeclareMathOperator{\lm}{LM}

\DeclareMathOperator{\LP}{LP}

\newcommand{\arash}[1]{\noindent{\bf {\color{blue!60!black}{\sc Arash:}  #1}}}

\newcommand{\alantha}[1]{\noindent{\bf {\color{magenta!80!black}{\sc Alantha:}  #1}}}

\newenvironment{cproof}
{\begin{proof}
 [Proof.]
 \vspace{-1.5\parsep}
}
{ \end{proof}}

\begin{document}{\bibliographystyle{alpha}}

\title{Efficient constructions of convex combinations for 2-edge-connected subgraphs on
fundamental classes}

\author{\textsc{Arash Haddadan\thanks{University of Virginia,
      Charlottesville, VA, USA. {\tt{ahaddada@virginia.edu}}.}} \and \textsc{Alantha
    Newman\thanks{CNRS-Universit\'e Grenoble Alpes, 38000,
      France. {\tt{alantha.newman@grenoble-inp.fr}}.}}}

\maketitle

\begin{abstract}
We present coloring-based algorithms for tree augmentation and use
them to construct convex combinations of 2-edge-connected subgraphs.
This classic tool has been applied previously to the problem, but our
algorithms illustrate its flexibility, which---in coordination with
the choice of spanning tree---can be used to obtain various properties
(e.g., 2-vertex connectivity) that are useful in our applications.

We use these coloring algorithms to design approximation algorithms
for the 2-edge-connected multigraph problem (2ECM) and the
2-edge-connected spanning subgraph problem (2ECS) on two 
well-studied types of LP solutions.  The first type of points, half-integer
square points, belong to a class of {\em fundamental extreme points},
which exhibit the same integrality gap as the general case.  For
half-integer square points, the integrality gap for 2ECM is known to
be between $\frac{6}{5}$ and $\frac{4}{3}$.  We improve the upper
bound to $\frac{9}{7}$.  The second type of points we study are {\em
uniform points} whose support is a 3-edge-connected graph and each
entry is $\frac{2}{3}$.  Although the best-known upper bound on the
integrality gap of 2ECS for these points is less than $\frac{4}{3}$,
previous results do not yield an efficient algorithm.  We give the first
approximation algorithm for 2ECS with ratio below $\frac{4}{3}$ for
this class of points.

\end{abstract}

\section{Introduction}

Due, at least in part, to its similarities and connections to the
\textbf{traveling salesman problem (TSP)}, the
\textbf{2-edge-connected spanning multigraph problem (2ECM)} is a
well-studied problem in the areas of combinatorial optimization and
approximation algorithms, and the two problems have often been studied
alongside each other.  Let $G=(V,E)$ be a graph with edge weights $w
\in \mathbb{R}^E_+$.  TSP is the problem of finding a minimum weight
connected spanning Eulerian multigraph of $G$ (henceforth a
\textit{tour of $G$}).  Note that a tour is Eulerian and connected,
which implies that it is also 2-edge-connected. 2ECM is the problem of
finding a minimum weight 2-edge-connected spanning multigraph of $G$
(henceforth a \textit{2-edge-connected multigraph of $G$}) and is a
relaxation of TSP.  A well-studied relaxation for both TSP and 2ECM on
a graph $G=(V,E)$ is as follows.
\begin{align*}
\min& \quad wx\\
\text{subject to:}&\quad x(\delta(S))\geq 2& \text{for $\emptyset \subset S \subset V$}\\
&\quad x\geq 0.  
\end{align*}
Let $\HK(G)$ be the feasible region of this LP. The integrality gap
$\gap$ is defined as
\begin{equation*} \sup_{G,w}
  \frac{\min \{{wx \;:\; x \mbox{ is the incidence vector of a
        2-edge-connected multigraph of $G$}} \}}{\min\{wx\; :\;x\in
    \HK(G)\}}.
\end{equation*} 
Alternatively, $\gap$ is the smallest number such that for any graph
$G$ and for any $x\in \HK(G)$, vector $\gap x$ dominates
a convex combination of 2-edge-connected multigraphs of
$G$ (see \cite{Goemans95} or Theorem 1 of \cite{carrvempala}).  Wolsey's analysis of Christofides'
algorithm shows that $\gap \leq \frac{3}{2}$
\cite{christofides,wolsey1980heuristic}, which is currently the
best-known approximation factor for 2ECM.  This seems strange since
Christofides' algorithm finds tours, which are more constrained than
2-edge-connected multigraphs.

Stated as a potentially easier-to-prove variant of the famous
four-thirds conjecture for TSP, it has been conjectured that $\gap\leq
\frac{4}{3}$ (e.g., Conjecture 2 in \cite{Carr98}, Conjecture 1 in
\cite{alexander2006integrality} and Conjecture 4 in
\cite{Boyd-Carr-2011}).  However, in contrast to the four-thirds
conjecture, the largest lower bound only shows that $\gap \geq
\frac{6}{5}$ \cite{alexander2006integrality}. Based on this lower
bound and computational evidence, Alexander, Boyd and Elliott-Magwood
proposed the following stronger conjecture (Conjecture 6 in
\cite{alexander2006integrality}), to which we will refer as the
\textit{six-fifths conjecture}.

\begin{conjecture}\label{sixfifth}
	If $x\in \HK(G)$, then $\frac{6}{5}x$ dominates a convex
        combination of 2-edge-connected multigraphs of $G$. In other
        words, $\gap \leq \frac{6}{5}$.
\end{conjecture}
Despite lack of progress in the general case, there has been some
progress towards validating this conjecture for special cases.
For example, in the unweighted case (when $w_e = 1$ for all $e\in E$),
there has been some success in beating the factor of
$\frac{3}{2}$~\cite{sebo12, boyd2014frac,bit13,takazawa}.  Another
important special case is half-integer points, which are conjectured
to exhibit the largest gap for TSP (e.g., see \cite{schalekamp20142,
  boydsebo-mp}).  Carr and Ravi proved that $\gap\leq \frac{4}{3}$ if
the optimal solution to $\min_{x\in \HK(G)} wx$ is
half-integer~\cite{Carr98}.  More specifically, they proved that if
$x\in \HK(G)$ and $x_e$ is a multiple of $\frac{1}{2}$ for $e\in E$,
then $\frac{4}{3}x$ dominates a convex combination of 2-edge-connected
multigraphs of $G$.  Very recently, Boyd et al. turned this proof into
a polynomial-time algorithm~\cite{boyd_et_al:LIPIcs:2020:12664}.

In this paper, we focus on two well-studied special cases of feasible
solutions for $\HK(G)$ both of which, directly or indirectly, can be
viewed as special cases of half-integer points.  Let $G=(V,E)$ be a
3-edge-connected graph.  Observe that $\frac{2}{3}\chi^E \in
\HK(G)$.  Such solutions can be reduced to half-integer solutions (see
proof of Theorem \ref{uniformcover-cubic}) and
provide a lower bound on the integrality gap for half-integer
solutions (see Theorem 4.10 \cite{arash-phd2020}).  Conjecture
\ref{sixfifth} implies that $\frac{4}{5}\chi^E$ can be written as a
convex combination of 2-edge-connected multigraphs of $G$.  
In fact, this was proved by Boyd and Legault, who actually proved the
stronger statement that $\frac{4}{5}\chi^E$ dominates a convex
combination of 2-edge-connected {\em subgraphs} of $G$ (a subgraph has
at most one copy of each edge)~\cite{philip6/7}.  The factor of
$\frac{4}{5}$ was subsequently improved even further to
$\frac{7}{9}$~\cite{philip}.

However, these proofs do not yield polynomial-time (approximation)
algorithms for the 2ECS problem, which is the problem of finding a
minimum weight 2-edge-connected spanning {\em subgraph} of $G$
(henceforth a {\em 2-edge-connected subgraph of $G$}).  This
gives rise to the following natural problem: \textit{Find a small
  positive rational number $\alpha$ such that for a 3-edge-connected
  graph $G=(V,E)$, the vector $\alpha \chi^E$ dominates a convex
  combination of 2-edge-connected subgraphs of $G$ and this convex
  combination can be found \textbf{in polynomial time}.}  The
best-known answer to this question is
$\frac{8}{9}$~\cite{cheriyan19992,uniform}.  (If we allow
2-edge-connected multigraphs instead of 2-edge-connected subgraphs,
the best-known answer to this question is
$\frac{15}{17}$~\cite{uniform}.)  In this paper, we improve this
factor.

\begin{restatable}{theorem}{uniformCover}
	\label{uniformcover}
	Let $G=(V,E)$ be a 3-edge-connected graph. The vector 
        $\frac{7}{8}\chi^E$ dominates a convex combination of
        2-edge-connected subgraphs of $G$. Moreover, this convex
        combination can be found in time polynomial in the size of
        $G$.
\end{restatable}

We can improve this factor to $\frac{41}{47}$ (Theorem
\ref{doubled-edges}) when we are finding convex combination of
multigraphs (i.e., when we allow doubled edges).  One consequence of
Theorem \ref{uniformcover} is a polynomial-time algorithm to write
$(\frac{6}{5} + \frac{1}{120})x$ as a convex combination of
2-edge-connected multigraphs when $x$ is a half-integer triangle point
(see Theorem 5.1 in \cite{arash-phd2020}.)  Half-integer triangle
points were introduced in \cite{Boyd-Carr-2011} as a class of points
for which the conjectured lower bound of $\frac{4}{3}$ for integrality
gap of TSP with the subtour elimination relaxation is achieved. Later
\cite{alexander2006integrality} showed that this class also achieves
the conjectured lower bound of $\frac{6}{5}$ for
$\alpha^{\HK}_{\EC}$. Boyd and Legault \cite{philip6/7} showed that
$\frac{6}{5}$ is also an upper bound on $\alpha^{\HK}_{\EC}$ when
restricted to half-integer triangle points. However, their proof does not
yield an (efficient) approximation algorithm.  We also remark that
very recently, Boyd et al., gave an efficient algorithm to write
$\frac{7}{8} \chi^E$ as a convex combination of 2-edge-connected {\em
  multigraphs} when $G=(V,E)$ is
3-edge-connected~\cite{boyd_et_al:LIPIcs:2020:12664}.

Another approach to the six-fifths conjecture is to consider so-called
{\em fundamental extreme points} introduced by Carr and
Ravi~\cite{Carr98} and further developed by Boyd and
Carr~\cite{Boyd-Carr-2011}.  A \textit{Boyd-Carr point} is a point $x$
of $\HK(G)$ that satisfies the following conditions. 
  \begin{itemize}
  	\item[(i)] The support graph of $x$ is cubic and 3-edge-connected.
 	\item[(ii)] For each vertex, there is exactly one
 incident edge $e$ such that $x_e = 1$ (a $1$-edge.)
  	\item[(iii)] The fractional edges form disjoint 4-cycles. 
  \end{itemize}
Boyd and Carr proved that in order to bound $\gap$ (e.g., to prove the
six-fifths conjecture), it suffices to prove a bound for Boyd-Carr
points~\cite{Boyd-Carr-2011}.  A generalization of Boyd-Carr points
are \textit{square points}, which are obtained by replacing each
1-edge in a Boyd-Carr point by an arbitrary-length path of 1-edges.
Half-integer square points are particularly interesting for various
reasons.  For every $\epsilon>0$, there is a half-integer square point
$x$ such that $(\frac{6}{5}-\epsilon)x$ does {\em not} dominate a
convex combination of 2-edge-connected multigraphs in the support of
$x$.  In other words, the lower bound for $\gap$ is achieved for
half-integer square points.  (This specific square point is discussed
in Section \ref{hard-to-round}.)  Furthermore, half-integer square
points also demonstrate the lower bound of $\frac{4}{3}$ for the
integrality gap of TSP with respect to the Held-Karp
relaxation~\cite{boydsebo-mp}.

Recently, Boyd and Seb{\H{o}} initiated the study of improving upper
bounds on the integrality gap for these classes and presented a
$\frac{10}{7}$-approximation algorithm (and upper bound on the integrality
gap) for TSP in the special case of {\em half-integer square points}.  They
pointed out that, despite their significance, not much effort has been
expended on improving bounds on the integrality gaps for these classes
of extreme point solutions.

In this paper, we focus on this class of solutions and we
improve the best-known upper bound
on $\gap$ for half-integer square points.  The best previously-known
upper bound on $\gap$ for half-integer square points is $\frac{4}{3}$,
which follows from the aforementioned bound of Carr and Ravi on
all half-integer points~\cite{Carr98}.  We note that there is also a
simple $\frac{4}{3}$-approximation algorithm using the observation
from \cite{boydsebo-mp} that the support of a square point is
Hamiltonian.  Our main result is to improve this factor.
\begin{restatable}{theorem}{squarepoint}
\label{square-point}
Let $x$ be a half-integer square point. Then $\frac{9}{7}x$ dominates
a convex combination of 2-edge-connected multigraphs in $G_x$, the
support graph of $x$.  Moreover, this convex combination can be found
in polynomial time.
\end{restatable}

\subsection{Overview of our methods}

A common approach to TSP and 2ECM is to choose a spanning tree that is
specially tailored to a particular type of cheap augmentation (e.g.,
Gao trees for path-TSP~\cite{gao2013lp}, max-entropy spanning trees
for asymmetric TSP~\cite{asadpour2017log} and rainbow 1-trees for
TSP~\cite{boydsebo-mp}).  Our algorithms fall within this general
paradigm.  To construct our spanning trees and augmentations, we use
the following three key ingredients: (i) gluing solutions over 3-edge
cuts, (ii) rainbow spanning tree decompositions and (iii) the top-down
coloring framework for tree augmentation.  The first ingredient,
gluing solutions for the 2-edge-connected subgraph problem over 3-edge
cuts, was introduced by Carr and Ravi~$\frac{4}{3}$~\cite{Carr98} and
has by now become a standard tool~\cite{philip6/7,philip}.  In each of
these works, a common pitfall is that the number of times the gluing
procedure is applied is not provably polynomial, leading to
inefficient algorithms.  We take a different approach to ensure a
polynomial running time.  While we do use gluing in the proof of
Theorem \ref{uniformcover}, we use it more sparingly (i.e., only over
proper 3-edge cuts that appear in the initial input graph and
therefore only a polynomial number of times) and with the sole
objective of assuming that the input graph is essentially
4-edge-connected.

The second ingredient is rainbow spanning tree decompositions, which
were introduced by Broersma and Li~\cite{broersma1997spanning} and
recently employed by Boyd and Seb\H{o} who used it to control the
parity of cuts of spanning trees over certain 4-edge cuts in
half-integer square points~\cite{boydsebo-mp}.  In this paper, we
apply this tool in a new way and for a completely different purpose
(unrelated to parity), highlighting its flexibility and usefulness.
Roughly speaking, this tool allows the edges of a graph to be
partitioned (subject to certain constraints) so that only one edge
from each set in the partition belongs to a spanning tree.
Hypothetically, this powerful decomposition tool could be applied to
control properties of the spanning trees output in the convex
combination of $x \in \HK(G)$ in, for example, an implementation of
the best-of-many Christofides' algorithm.  However, exactly which
properties can be obtained and how to use such properties is not yet
clear; the power of this tool is likely still far from being fully
realized.  In this paper, we use it to obtain spanning trees in which
few pairs of adjacent vertices in the graph are both leaves of the
tree.

The third ingredient is coloring-based algorithms for tree
augmentation.  Such algorithms were recently
studied by Iglesias and Ravi who introduced a top-down coloring
algorithm for tree augmentation~\cite{IR}, which generalized results of
Cheriyan, Jordan and Ravi~\cite{cheriyan19992}.  A straightforward
application of the algorithm of Iglesias and Ravi can be used to prove Theorem
\ref{uniformcover} with a factor $\frac{8}{9}$.  In this paper, we significantly
extend this coloring-based framework.  We prove Theorem \ref{uniformcover} by
demonstrating that certain properties of the graph (e.g., essentially
4-edge-connected) and the spanning tree (e.g., the leaf structure
obtained via rainbow decomposition) can be used to design more careful
coloring rules which results in smaller tree augmentations.  In a key
step in the proof of Theorem \ref{square-point}, we tailor the
coloring rules to obtain a convex combination of 2-vertex-connected
subgraphs with minimum degree three.  The objective here is to
construct convex combinations which use the half-edges in the
half-integer square points sparingly.  The property of 2-vertex
connectivity and the fact that the complement of a subgraph forms a
matching crucially allows us to be more parsimonious with the
half-edges when constructing the convex combinations.  Thus, we
demonstrate that this coloring framework is a flexible and therefore
powerful tool in designing approximation algorithms for 2ECM and
related problems, and we believe it likely has further applications.

\section{Preliminaries and tools}\label{sec:tools}
In this paper, a graph may contain parallel edges.
We work with multisets of edges of $G=(V,E)$. For a multiset $F$ of
$E$, the submultigraph induced by $F$ (henceforth, we simply call $F$
a \textit{multigraph} of $G$) is the graph with the same number of
copies of each edge as $F$. A subgraph of $G$ has at most one copy of
each edge in $E$.  The incidence vector of multigraph $F$ of $G$,
denoted by $\chi^{F}$ is a vector in $\mathbb{R}^E$, where
$\chi^{F}_e$ is the number of copies of $e$ in $F$. For multigraphs
$F$ and $F'$ of $E$, we define $F+F'$ to be the multigraph that
contains $\chi^{F}_e+\chi^{F'}_e$ copies of each edge $e\in E$.  For a
subset $S \subset V$ of vertices, let $\delta(S)$ be the edges in $E$
with one endpoint in $S$ and one endpoint not in $S$. For a subgraph
$F$ of $G$, $\delta_F(S)=\delta(S)\cap F$.

If multigraph $F$ spans $V$ and is 2-edge-connected, we say $F$ is a
\textit{2-edge-connected spanning multigraph} of $G$ (or a
2-edge-connected multigraph of $G$ for brevity). If in addition, $F$
is a subgraph, we say $F$ is a 2-edge-connected subgraph of $G$.  Let
$\mathcal{F}$ be a collection of multigraphs of $G$.  We say
$\sum_{i=1}^k \lambda_i \chi^{F_i}$ is a {\em convex combination} of
$k$ multigraphs in $\mathcal{F}$ if $\lambda_i > 0$ for all $i \in
\{1,\ldots,k\}$ and $\sum_{i=1}^k \lambda_i = 1$.  Let $x$ be a vector
in $\mathbb{R}^E$.  We say {\emph{$x$ can be written as convex
    combination of multigraphs in $\mathcal{F}$}} if additionally $x =
\sum_{i=1}^k \lambda_i \chi^{F_i}$, and $\lambda_i$ and $F_i$ for $i
\in \{1,\ldots,k\}$ can be found in polynomial time in the size of
$x$.  Here by the size of $x$ we refer to $|E|$ (i.e., the number of
edges in the support of $x$).  For a vector $x,y\in \mathbb{R}^E$ we
say $y$ \textit{dominates} $x$ if $y_e\geq x_e$ for $e\in E$.  For any
$ x\leq y$, if $x$ can be written as a convex combination of
multigraphs in $\mathcal{F}$, then we say {\em $y$ dominates a convex
  combination of multigraphs in $\mathcal{F}$}.  Moreover, for $x \leq
y \leq 1$ ($x \leq y \leq 2$, respectively), if $x$ can be written as
a convex combination of 2-edge-connected subgraphs (multigraphs,
respectively), then $y$ can be written as a convex combination of
2-edge-connected subgraphs (multigraphs, respectively).  Thus, in our
proofs we sometimes use the phrase ``dominates a convex combination''
in place of ``written as a convex combination'' when the context is
appropriate.

 The support graph of $x$, denoted by $G_x$ is the graph induced on
 $G$ by $E_x = \{e\in E\; : \; x_e >0\}$. Vector $x$ is
 \textit{half-integer} if $x_e$ is a multiple of $\frac{1}{2}$ for all
 $e\in E$.  We say that an edge $e\in G_x$ is a \textit{1-edge} if
 $x_e=1$. Similarly, an edge is a \textit{half-edge} if $x_e =
 \frac{1}{2}$.

Finally, we say graph $G$ is {\em $k$-edge-connected} if
$|\delta(S)|\geq k$ for all $\emptyset\subset S\subset V$, and is {\em
  essentially $k+1$-edge-connected} if additionally $|\delta(S)|\geq
k+1$ for all $S\subset V$ with $2\leq |S| \leq |V|-2$.

Next, we introduce some key tools.

\subsection{Cycle covers}\label{tools:cc}

A \textit{cycle cover} $\C$ of graph $G=(V,E)$ is a subgraph of $G$
where every vertex $v\in V$ belongs to exactly one cycle in $\C$.  
We now present some (well-known) observations pertaining to cycle
covers that we will use.

\begin{observation}\label{obs:cycle-cover}
Let $G = (V,E)$ be a 2-edge-connected cubic graph.  Then vector
$\frac{2}{3}\chi^E$ can be written as a convex combination of cycle
covers of $G$ where each cycle cover covers the 3-edge cuts of $G$.
\end{observation}

\begin{proof}
The vector $\frac{1}{3}\chi^E$ belongs to the perfect matching
polytope (see Theorem 4 in \cite{naddef1981matchings}).  Thus, this
vector can be written as the convex combination of at most $|E|+1$
perfect matchings via Carath\'eodory's
Theorem~\cite{caratheodory1911variabilitatsbereich}.  Moreover, each
of these perfect matchings intersects each 3-edge cut of $G$ in
exactly one edge.  (This fact has been used many times and can be
considered folklore; see \cite{kaiser2005unions} for an application.)
In a cubic graph, the complement of a perfect matching is a cycle
cover.  This proves the observation.
\end{proof}

Notice that when we say a cycle cover {\em covers} a 3-edge cut, we
mean that the cycle cover intersects this cut in exactly two edges.
One useful property of a cycle cover that covers all 3-edge cuts of a
graph is that it does not contain any cycle with length less than four
(i.e., it does not contain a triangle).

\begin{observation}\label{obs:in-subtour}
	Let $G=(V,E)$ be 3-edge-connected cubic graph and
        $\mathcal{C}$ be a cycle cover of $G$ that covers all 3-edge
        cuts of $G$. Then the vector $x = \chi^E - \frac{1}{2}
        \chi^{\C}$ belongs to
        $\HK(G)$.  Moreover, $x(\delta(v)) = 2$ for all $v \in V$.
\end{observation}

\begin{proof}
	Take $\emptyset\subset U \subset V$.  If $|\delta(U)|\geq 4$,
        then clearly $x(\delta(U))\geq 2$.  Otherwise, $|\delta(U)| =
        3$.  In this case, since exactly two edges in $\delta(U)$ belong to
        $\mathcal{C}$, there is one edge $e\in \delta(U)$
        with $x_e=1$. Hence, $x(\delta(U)) = 2$. Therefore, $x \in
        \HK(G)$.
\end{proof}

\begin{theorem}[\cite{bit13}]\label{bit}
		Let $G=(V,E)$ be a 3-edge-connected cubic graph. We
                can in polynomial time find a cycle cover
                $\mathcal{C}$ of $G$ such that for every $S\subseteq
                V$ for which $|\delta(S)|\leq 4$, we have
                $\mathcal{C}\cap \delta(S) \neq \emptyset$ (i.e.,
                cycle cover $\mathcal{C}$ covers all 3-edge cuts and
                4-edge cuts of $G$).
	\end{theorem}

\subsection{Gluing over 3-edge cuts}\label{gluingtool}

A tool introduced by Carr and Ravi and frequently used for
constructing convex combinations of 2-edge-connected subgraphs is {\em
  gluing solutions over 3-edge cuts}~\cite{Carr98}.  This allows us to
focus on essentially 4-edge-connected graphs, as stated in the next
lemma whose full proof can be found in Appendix \ref{app:gluing}.

\begin{restatable}{theorem}{reductionToEssentiallyFourEC}\label{reductionToFourEC}
For $\alpha \in [0,1]$, the following two statements are equivalent.

\begin{enumerate}
\item For an essentially 4-edge-connected graph $G=(V,E)$ and any cycle
  cover $\C$ of $G$, the vector $\chi^E - \alpha \chi^{\C}$ can be written as a
  convex combination of 2-edge-connected subgraphs of $G$.

\item For a 3-edge-connected graph $G=(V,E)$ and any cycle
  cover $\C$ of $G$ that covers all the 3-edge cuts of $G$, the vector
  $\chi^E - \alpha \chi^{\C}$ can be written as a
  convex combination of 2-edge-connected subgraphs of $G$.

\end{enumerate}
\end{restatable}

\subsection{Rainbow 1-tree decomposition}\label{tools:rainbow}

Given a graph $G=(V,E)$, a {\em 1-tree} $T$ of $G$ is a connected
spanning subgraph of $G$ containing $|V|$ edges, where the vertex
labeled 1 has degree exactly two and $T\setminus{\delta(1)}$ is a
spanning tree on $V\setminus \{1\}$.  Boyd and Seb\H{o} proved the
following theorem (see Theorem 5 in \cite{boydsebo-mp}).  In fact,
they showed that the relevant decomposition can be found in time
polynomial in the size of graph $G$.

\begin{theorem}[\cite{boydsebo-mp}]\label{boydsebo-rainbow}
	Let $x\in \HK(G)$ be half-integer, $x(\delta(v))=2$ for all $v\in
        V$, and $\mathcal{P}$ be a partition of the half-edges into
        pairs. Then $x$ can be written as a convex combination of
        1-trees of $G$ such that each 1-tree contains exactly one edge
        from each pair in $\mathcal{P}$.
\end{theorem}

To prove Theorem \ref{uniformcover}, we construct 2-edge connected
subgraphs by augmenting 1-trees.  However, it is often easier to think
of augmenting spanning trees.  We will use the term {\em connector} to
refer to a subgraph of $G$ that is either a spanning tree or a 1-tree
of $G$.  For a 1-tree $K$, let $C_K$ denote the unique cycle in $K$.
Note that $K/C_K$ is a spanning tree $T$ of $G/C_K$ that is obtained when
we contract that cycle to a single vertex.  We let the vertex
corresponding to the contracted $C_K$ be the root $r$ of $T$.  In a
connector, each vertex with degree 1 is a {\em leaf}.  The exception
to this is that the root of a spanning tree is not a leaf (even if it
has degree 1).  We have the following useful facts.

\begin{observation}\label{obs:connector2tree}
Let $G=(V,E)$ be a graph, let $K$ be a 1-tree of $G$ with cycle $C_K$
and let $T= K/C_K$ be a spanning tree with root $r$ (where $r$
corresponds to the contracted $C_K$).

\begin{enumerate}
\item Let $F \subset E(G/C_K) \setminus{T}$ and suppose that $T + F$
is 2-edge connected.  Then $K + F$ is 2-edge connected.

\item A vertex $v \in V$ is a leaf in $K$ iff it is a leaf in $T$.
\end{enumerate}
\end{observation}

\subsection{Tree augmentation and the top-down coloring
  framework}\label{subsec:coloring}

We now describe the {\em top-down coloring} framework which is key to
proving both of our main results.  Consider a graph
$G=(V,E)$ and a spanning tree $T$ of $G$. Let $L=E\setminus T$ be the
set of links, and let $w\in \mathbb{R}^L_{\geq 0}$ be a weight
vector. The tree augmentation problem asks for the minimum weight
$F\subseteq L$ such that $T+F$ is 2-edge-connected (i.e., $F$ is a
{\em feasible augmentation for $T$}).  For a link $\ell\in L$, denote
by $P_\ell$ the unique path between the endpoints of $\ell$ in
$T$. For an edge $e\in T$, we say $\ell \in \cov(e)$ if $e\in P_\ell$.

For $q \in \Zplus$ and vector $\p\in \Zplus^L$, where $\p =(p_1,
\ldots, p_{|L|})$ and $\pel \leq q$ for all $\ell \in L$, a {\em
  $(\p,q)$-coloring of $L$} is a function $\gamma : L \rightarrow
2^{\{c_1,\ldots,c_q\}}$ where $\gamma(\ell)$ has size at most $\pel$
for $\ell \in L$.  In other words, a $(\p,q)$-coloring of $L$ assigns at most
$\pel$ different colors from a set of $q$ available colors to each link
$\ell \in L$.  Although $\gamma$ is defined to be a function on $L$,
we abuse notation and for $e \in T$, we let $\rec(e) = \bigcup_{\ell
  \in cov(e)} \gamma(\ell)$ denote the set of (distinct) colors edge
$e \in T$ has received in the coloring $\gamma$.  For a
$(\p,q)$-coloring $\gamma$ of $L$, an edge $e \in T$ and $i \in \{1,
\ldots, q\}$, we say {\em $e$ received color $c_i$} if $c_i\in
\gamma(e)$, otherwise, we say {\em $e$ is missing color $c_i$}. We
denote the set of colors an edge $e$ is missing by $\bar{\gamma}(e)$
which is defined by $\{c_1,\ldots,c_q\}\setminus \gamma(e)$.

\begin{definition}
Let $\gamma$ be a $(\p,q)$-coloring of $L$. We say $\gamma$ is {\em
  $T$-admissible} if for each edge $e \in T$, $e$ has received all $q$
colors (i.e., for each $e \in T$, we have
$\rec(e) =\{c_1,\ldots,c_q\}$).
\end{definition}

We note that Observations \ref{combination}, \ref{firstlink} and
\ref{ansectory} are from \cite{IR}, but we present them here
using our notation.  

\begin{observation}\label{combination}
	Let $T$ be a spanning tree of $G=(V,E)$ and $L =
        E\setminus{T}$ be the set of links. If there exists a
        $T$-admissible $(\p,q)$-coloring of $L$, namely $\gamma$, then
        the vector $z\in \mathbb{R}^L$, where
        $z_\ell=\frac{\pel}{q}$ for $\ell\in L$, dominates a convex
        combination of feasible augmentations of $T$.  Moreover, given
        $\gamma$, this convex combination can be found in polynomial
        time.
\end{observation}
\begin{proof}
For $i\in \{1, \ldots, q\}$, let $A_i = \{\ell \in L : c_i \in
\gamma(\ell)\}$. By the definition of $T$-admissibility, for each
$e\in T$ and each color $c_i\in \{c_1,\ldots,c_q\}$ there is at least
one link $\ell\in L\cap \cov(e)$ with $c_i\in \gamma(\ell)$. Hence,
for each $i \in \{1, \ldots,q\}$, $A_i$ is a feasible augmentation for
$T$. Moreover, a link $\ell$ is in at most $\pel$ of
$\{A_1,\ldots,A_q\}$ since a link $\ell$ is colored with at most
$\pel$ colors.  Finally, observe that $\sum_{i=1}^{q} \frac{1}{q}
\chi^{A_i} \leq z$.
\end{proof}

Now we are almost ready to define a {\em top-down $(\p,q)$-coloring
  algorithm} for finding a $T$-admissible $(\p,q)$-coloring of the
links $L$.  We first need to introduce some more terminology.  If we
choose a vertex $r\in V$ to be the root of tree $T$, we can think of
$T$ as an arborescence, with all edges oriented away from the root.
For a link $\ell=uv$ in $L$, a {\em least common ancestor} of $u$ and
$v$, denoted by $LCA(\ell)$, is the vertex $s$ that has edge-disjoint
directed paths to $u$ and $v$ in $T$.  An edge $e \in T$
  is an ancestor of $f \in T$ if there is a directed path containing
  both $e$ and $f$.  (Note that $e$ is an ancestor of itself.)  By
{\em LCA order}, we mean the partial ordering of the links according
to their LCAs (i.e., if $LCA(\ell)$ is higher than $LCA(\ell')$, then
$\ell < \ell'$ in the partial order).  For a link $\ell = uv$ where $s
= LCA(\ell)$, we use $\Lu$ to denote the edges in $T$ on the path from
$s$ to $u$ and $\Rv$ for the edges in $T$ on the path from $s$ to $v$.

 In each iteration of a top-down $(\p,q)$-coloring algorithm, we choose
 a link $\ell\in L$ and color $\ell$ with $\pel$ different colors from a
 set of $q$ available colors, $\{c_1,\ldots,c_q\}$.  There are two key
 requirements: (i) the links are chosen in any order that respects the
 LCA order, and (ii) we continue until all links are colored.

After each iteration of a top-down $(\p,q)$-coloring algorithm, we
have a $(\p,q)$-coloring of the links, and if some links are not yet
colored, we sometimes refer to this as a {\em partial
  $(\p,q)$-coloring} of the links.  
For a partial coloring of the
links, $e \in T$ and $ i \in \{1,\ldots, q\}$, we say color $c_i$ is
{\em missing} for edge $e$ if no links in $\cov(e)$ have been colored
with $c_i$.
When coloring
link $\ell$ with $c_i$ as one of its $\pel$ colors, we say $e$ {\em
  receives a new color $c_i$} or {\em color $c_i$ is new for edge $e$}
if $e\in P_\ell$ and edge $e$ was missing $c_i$ before this iteration
of the algorithm.

\begin{observation}\label{firstlink}
Consider a partial $(\p,q)$-coloring $\gamma$ of $L$ produced after
some iterations of a top-down $(\p,q)$-coloring algorithm.  For an
edge $e \in T$, let $\ell$ be any link in $cov(e)$.  If $\ell$ is colored in $\gamma$, then $e$ has received at least
$\pel$ colors (i.e., $|\rec(e)|\geq \pel$).
\end{observation}

\begin{observation}\label{ansectory}
Consider a partial $(\p,q)$-coloring $\gamma$ of $L$ produced after
some iterations of a top-down $(\p,q)$-coloring algorithm.  Consider $e, f
\in T$ where $e$ is an ancestor of $f$.  For any link $\ell \in \cov(e)
\cap \cov(f)$ that is not colored in $\gamma$, 
any color given to $\ell$ that is new
for $e$ is also new for $f$.
\end{observation}

Now consider a partial $(\p,q)$-coloring $\gamma$ of $L$ produced
after some iterations of a top-down $(\p,q)$-coloring algorithm.  Let
$\ell \in L$ be a link that is not yet colored in $\gamma$ and let $P
\in \{\Rv, \Lu\}$.  Then $P=e_1,\ldots,e_k$ where $e_k$ is the
lowest edge in the path $P$.  By Observation \ref{ansectory}, we have
$\gamma(e_k)\subseteq \gamma(e_{k-1})\subseteq \ldots \subseteq
\gamma(e_1)$, or $\bar\gamma(e_1)\subseteq \ldots \subseteq
\bar\gamma(e_{k})$.  We define \textit{the top $i$ missing colors for
  path $P$}, denoted by $\bar\gamma_i(P)$, to be a set of size $i$
where $\bar\gamma(e_{j-1})\subset \bar\gamma_i(P)\subseteq
\bar\gamma(e_j)$ where $e_j$ is the highest edge in $P$ with
$|\bar\gamma(e_j)|\geq i$.  If no such $e_j$ exists, it must be that
$|\bar\gamma_i(e_k)| < i$ and we define $\bar\gamma_i(P)$ to be
$\bar\gamma(e_k)$ (in which case $|\bar\gamma_i(P)| < i$).

A $(\p,q)$-coloring algorithm is {\em $T$-admissible} if the final
$(\p,q)$-coloring of $L$ (i.e., after the last iteration) is
$T$-admissible. Throughout this paper, we prove that a top-down
coloring algorithm is $T$-admissible by showing that after the iteration in
the algorithm when all the links covering an edge $e \in T$ are
colored, the edge $e$
has received all $q$ colors.

\subsubsection{A simple application of the top-down coloring algorithm}

To illustrate the utility of the top-down coloring framework, we show
how it can be used to state a short proof of a theorem of DeVos,
Johnson and Seymour~\cite{devos2003cut}.  Here, the key fact is that
for each spanning tree $T$, a $T$-admissible top-down $(\p,q)$-coloring
algorithm produces only $q$ feasible augmentations.

\begin{theorem}[\cite{devos2003cut}]\label{devosJohnsonSeymour}
Let $G=(V,E)$ be a 3-edge-connected graph.  Then there exists a
partition of $E$ into sets $\{X_1, X_2, \ldots, X_9\}$ (where $X_i$ is
allowed to be empty) such that the graph $G_i = (V, E\setminus{X_i})$
is 2-edge-connected for $i \in \{1, \ldots, 9\}$.
\end{theorem}

Before we can prove Theorem \ref{devosJohnsonSeymour}, we need to
prove the following claim, which directly follows from \cite{IR}.

\begin{claim}\label{easy-color}
Let $G=(V,E)$ be a 3-edge-connected graph, let $T$ be a spanning
tree of $G$ with root $r$, and let $L = E \setminus{T}$.  Then there
is a $T$-admissible top-down $(2\chi^L,3)$-coloring algorithm to color the
links in $L$.
\end{claim}

\begin{proof}
Let $\gamma$ be the $(2 \chi^L,3)$-coloring of $L$ we maintain. At the start,
we have $\gamma(\ell)=\emptyset$ for all $\ell \in L$. To color link
$\ell = uv$, we use the following coloring rule.

\paragraph{Coloring Rule:} Give link $\ell$ colors $\bar\gamma_1(\Lu)$
and $\bar\gamma_1(\Rv)$. If either is empty, give $\ell$ an arbitrary
color that $\ell$ does not already have.

\vspace{5pt}

We now prove that this top-down coloring algorithm is
$T$-admissible. Consider an $e \in T$. If $e$ is an edge in $T$, then
since the graph is 3-edge-connected we have $|\cov(\ell)|\geq 2$.  Let
$\ell_1,\ell_2$ be two of the links in $\cov(e)$ with the highest
LCAs.

When coloring $\ell_1$, edge $e$ receives two new
colors by Observation \ref{firstlink}. Now consider
the iteration in which the algorithm colors $\ell_2$.  At
the time of coloring $\ell_2$, the top-down coloring algorithm that we 
described above will give $\ell_2$ at least one color that an ancestor
of $e$ is missing since $e$ is either in $\mathcal{R}_{\ell_2}$ or
$\mathcal{L}_{\ell_2}$.  By Observation \ref{ansectory}, we can conclude
that $e$ receives a new color after coloring $\ell_2$.  Thus, after we
have colored link $\ell_2$, edge $e$ has received at least $2+1=3$ colors.
\end{proof}

\begin{proof}[Proof of Theorem \ref{devosJohnsonSeymour}]
From the theorem of Nash-Williams~\cite{nw}, we know that $2G$
contains three edge-disjoint spanning trees of $G$.  Call these trees
$T_1, T_2$ and $T_3$.  Observe that each edge in $E$ is absent from at
least one of the three spanning trees.  For each $i \in \{1,2,3\}$, we
want to show that there is an admissible top-down
$(2\chi^L,3)$-coloring algorithm for $T_i$ and $L_i = E
\setminus{T_i}$.  Since $G$ is 3-edge-connected, we can apply Claim
\ref{easy-color}.  Observe that each link receives two colors and the
algorithm uses three colors in total.

For each $i \in \{1,2,3\}$, we obtain three augmentations $A^j_i
\subset L_i$ for $j \in \{1,2,3\}$ such that $A^j_i \cup T_i$ is
2-edge-connected.  The set $A^j_i$ contains all links in $L_i$ that
received color $j$ as one of their two colors.  Let $X^j_i = L_i
\setminus{A^j_i}$ be the set of links in $L_i$ that did not receive
color $j$.  Then for each $e \in L_i$, $e$ belongs to $X^j_i$ for some
$j \in \{1,2,3\}$.  Since each edge $e \in E$ belongs to $L_i$ for
some $i \in \{1,2,3\}$, we conclude that each edge $e \in E$ belongs
to at least one of the nine sets $X^j_i$ for $i,j \in \{1,2,3\}$.
\end{proof}

The top-down coloring framework might have further applications for
problems in which the objective is to obtain a convex combination of
few subgraphs.  Such problems were recently explored by H\"orsch and
Szigeti~\cite{horscheulerian}.

\section{Uniform cover for 2-edge-connected subgraphs}\label{sec:uniform}

Boyd and Legault \cite{philip6/7} showed that to prove
Theorem \ref{uniformcover}, it suffices to prove it for all cubic
3-edge-connected graphs (See Lemma 2.2 of \cite{philip6/7}). 

\begin{theorem}\label{uniformcover-cubic}
	Let $G=(V,E)$ be a 3-edge-connected cubic graph. Then
	$\frac{7}{8} \chi^E$ can be written as a convex combination of
	2-edge-connected subgraphs of $G$. 
\end{theorem}

Recall that we say a vector can be {\em written as a
  convex combination of subgraphs} if the convex multipliers and the
respective subgraphs can be constructed in polynomial time.  In order
to prove Theorem \ref{uniformcover-cubic} we prove the following
theorem.

\begin{theorem}\label{ccpoint-cubic}
	Let $G=(V,E)$ be a 3-edge-connected cubic graph and let
        $\mathcal{C}$ be a cycle cover of $G$ covering all 3-edge cuts
        of $G$. The vector $\chi^E - \frac{3}{16}\chi^{\mathcal{C}}$
        can be written as a convex combination of 2-edge-connected
        subgraphs of $G$.
\end{theorem}

\begin{proof}[Proof of Theorem \ref{uniformcover-cubic}]
Let $\sum_{i=1}^{k}\lambda_i \chi^{\mathcal{C}_i} = \frac{2}{3}
\chi^E$ be the convex combination of cycle covers obtained via
Observation \ref{obs:cycle-cover}. By Theorem \ref{ccpoint-cubic}, for
each $i\in\{1,\ldots,k\}$ we can write $\chi^E -
\frac{3}{16}\chi^{\C_i}$ as convex combination of 2-edge-connected
subgraphs of $G$. Hence, $\sum_{i=1}^{k} \lambda_i (\chi^E -
\frac{3}{16}\chi^{\C_i}) = \chi^E - \frac{1}{8}\chi^E =
\frac{7}{8}\chi^E$, and we conclude that $\frac{7}{8} \chi^E$ can be
written as a convex combination of 2-edge-connected subgraphs of $G$.
\end{proof}

Furthermore, applying Theorem \ref{reductionToFourEC} we can focus on
proving the next lemma, which implies Theorem \ref{ccpoint-cubic}.

\begin{lemma}\label{1-13/16}
Let $G=(V,E)$ be an essentially 4-edge-connected cubic graph and
$\mathcal{C}$ be a cycle cover of $G$. The vector $\chi^E -
\frac{3}{16}\chi^{\mathcal{C}}$ can be written as a convex combination
of 2-edge-connected subgraphs of $G$.
\end{lemma}

Our approach to proving Lemma \ref{1-13/16} is based on the top-down
coloring framework introduced in Section \ref{subsec:coloring}.  This
allows us to avoid gluing completely when dealing with an essentially
4-edge-connected cubic graph (in contrast to \cite{philip6/7},
\cite{philip}).  In particular, in an essentially 4-edge-connected
graph, if we consider any spanning tree $T$, then any edge $e\in T$
that is not incident to a leaf vertex is covered by at least three
links (i.e., $|\cov(e)|\geq 3$), as opposed to only two links if the
graph is only 3-edge-connected.  Therefore, assigning fewer colors to
each link still satisfies the requirements of the top-down coloring
algorithm for most of the edges in $T$.  The problematic links are
those that are incident to two leaves, since we cannot satisfy the
color requirements of both adjacent tree edges using fewer colors on
these links.  These problematic links (called {\em leaf-matching
  links}) must be assigned more colors.  Using a specially designed
rainbow 1-tree decomposition, we can ensure that there are actually
few such problematic links.  We now present some necessary
definitions.

\begin{definition}
Let $K$ be a connector of a graph $G=(V,E)$ and let $L =
E\setminus{K}$ denote the set of {\em links}.  We say an edge $e=uv
\in L$ is a {\em leaf-matching link} for $K$ if both $u$ and $v$ are
leaves in $K$.  We denote by $\lm(G,K)$ the set of leaf-matching links
for $K$ in $G$.
\end{definition}

The following lemma shows that using the top-down coloring algorithm
we can find feasible augmentations that are ``cheap'' when there are
few leaf-matching links.

\begin{restatable}{lemma}{topdownFewLM}\label{lem:newlemma}
	Let $H=(U,F)$ be an essentially 4-edge-connected graph and let
        $T$ be a spanning tree of $H$ with root $r$. Then we can find
        a $(\chi^{\lm(H,T)}+3\chi^{L},5)$-coloring of $L = F\setminus
        T$ in polynomial time.
	\end{restatable}

Lemma \ref{lem:newlemma} is not strong enough to prove Lemma
\ref{1-13/16}, but its proof helps illustrate our tools and
techniques.  The next lemma states that we can in fact find a convex
combination of 1-trees with few leaf-matching links (i.e., each edge
in $\C$ is a leaf-matching link in at most $\frac{1}{10}$ fraction of
the convex combination).

\begin{restatable}{lemma}{rainbowDecom} \label{rainbow2}
	Let $G=(V,E)$ be a 3-edge-connected cubic graph and let $\C$
        be a cycle cover of $G$.  The vector $\chi^E -
        \frac{1}{2}\chi^{\C}$ can be written as a convex combination
        of 1-trees $\{K_1, \dots, K_{k}\}$ of $G$.  Moreover, this
        convex combination (i.e., $\sum_{i=1}^k \lambda_i \chi^{K_i}$
        where $\lambda_i \in [0,1]$ and $\sum_{i=1}^k \lambda_i = 1$)
        has the following two properties: (i) $\sum_{i:e\in
          \lm(G,K_i)} \lambda_i \leq \frac{1}{10}$ for $e\in \C$, (ii)
        the links in $\lm(G,K_i)$ are vertex-disjoint.
\end{restatable}

We show how to use property (i) from Lemma \ref{rainbow2} to obtain a
convex combination with a weaker bound than that proved in Lemma
\ref{1-13/16}.  For an essentially 4-edge-connected cubic graph
$G=(V,E)$, let ${\cal C}$ be a cycle cover and let $y= \chi^E -
\frac{1}{2}\chi^{\C}$.  Then we show that $\chi^E -
\frac{9}{50}\chi^{\C}$ can be written as a convex combination of
2-edge-connected subgraphs of $G$.  Let $\sum_{i=1}^{k}\lambda_i
\chi^{K_i}$ be the convex combination of $y$ obtained via Lemma
\ref{rainbow2}. For $i \in \{1,\ldots,k\}$, let $H_i = G/C_{K_i} =
(U_i,F_i)$, $T_i = K_i/C_{K_i}$ and define
$\p^i=\chi^{\lm(H_i,T_i)}+3\chi^{F_i\setminus T_i}$. Lemma
\ref{lem:newlemma} implies that we can find a $(\p^i,5)$-coloring of
$F_i\setminus T_i$ for $i\in\{1,\ldots,k\}$. By Observation
\ref{combination}, we have $\frac{\p^i}{5}= \sum_{j=1}^{5}\frac{1}{5}
\chi^{A^i_j}$ where $A^i_j$ is a feasible augmentation for $T_i$ and
hence also for $K_i$.

Let $z^i_j = \chi^{K_i + A_j^i} $ be the characteristic vector of the
corresponding 2-edge-connected subgraph of $G$. Then, we have
\begin{eqnarray*}
z^i & = & \sum_{j=1}^5\frac{1}{5} z^i_j = \sum_{j=1}^5 \frac{1}{5}
\chi^{K_i + A_j^i} = \chi^{K_i} + \sum_{j=1}^5 \frac{1}{5} \chi^{A_j^i}\\ 
& = & \chi^{K_i} + \frac{3}{5} \chi^{L_i} +
\frac{1}{5}\chi^{\lm(H_i,T_i)}.
\end{eqnarray*}
Notice that $\sum_{i=1}^k \lambda_i \chi^{L_i} \leq
\frac{1}{2}\chi^{\cal{C}}$.  Moreover, by
property (i) of Lemma
\ref{rainbow2}, $\sum_{i=1}^k \lambda_i \chi^{\lm(H_i,T_i)} \leq
\frac{1}{10} \chi^{\C}$.
Next, we claim that the vector $z$ can be
written as a convex combination of 2-edge-connected subgraphs of $G$.
\begin{eqnarray*}
z  & = &  
\sum_{k=1}^k \lambda_i z^i = \sum_{i=1}^k\lambda_i \left(\chi^{K_i} + \frac{3}{5} \chi^{L_i} + \frac{1}{5}\chi^{\lm(H_i,T_i)}\right) \\
& = & \sum_{i=1}^k\lambda_i \chi^{K_i} + \frac{3}{5}\sum_{i=1}^k \lambda_i 
\chi^{L_i}  + 
\frac{1}{5}\sum_{i=1}^k\lambda_i \chi^{\lm(H_i,T_i)} \\
& \leq &
(\chi^E - \frac{1}{2} \chi^{\cal{C}}) + \frac{3}{10}\chi^{\cal{C}} + 
\frac{1}{50} \chi^{\C} = \chi^E - \frac{9}{50}\chi^{\C}.
\end{eqnarray*}

This is slightly worse than the factor promised in Lemma
\ref{1-13/16}. We show that by paying a bit more on non leaf-matching
links and exploiting a different property of the leaf-matching links,
namely that they can be vertex-disjoint (property (ii) in Lemma
\ref{rainbow2}), we can obtain the desired factor.

\begin{restatable}{lemma}{topdownfiveovereight}\label{5/8} 
	Let $H=(U,F)$ be an essentially 4-edge-connected graph
 and let
        $T$ be a spanning tree of $H$ with root $r$. If the edges in
        $\lm(H,T)$ are vertex-disjoint, then we can find a
        $T$-admissible $(5\chi^{L},8)$ coloring of $L = F\setminus{T}$
        in polynomial time.
\end{restatable}

The proof of Lemma \ref{5/8} is based on a top-down coloring algorithm
and is deferred to Section \ref{subsec:main-coloring-algs}. 

\begin{proof}[Proof of Lemma \ref{1-13/16}]
Let $y= \chi^E - \frac{1}{2}\chi^{\C}$. Let $\sum_{i=1}^{k}\lambda_i
\chi^{K_i}$ be the convex combination of vector $y$ obtained via Lemma
\ref{rainbow2}.  We now set $H_i = G/C_{K_i} = (U_i,F_i)$ and $T_i =
K_i/C_{K_i}$ (recall that $C_{K_i}$ is the unique cycle in the 1-tree
$K_i$).  By Lemma \ref{5/8}, we can find a $(5\chi^{F_i\setminus
  T_i},8)$-coloring of $L_i = F_i\setminus T_i$ for $i \in
\{1,\ldots,k\}$. By Observation \ref{combination}, we have
$\frac{5}{8}\chi^{L_i}= \sum_{j=1}^{8} \frac{1}{8} \chi^{A^i_j}$ where
$A^i_j$ is a feasible augmentation for $T_i$ and therefore for $K_i$
(by Observation \ref{obs:connector2tree}).  In other words, $K_i +
A_j^i$ for $i \in \{1, \ldots, k\}$ and $j \in \{1, \ldots, 8\}$ is a
2-edge-connected subgraph of $G$.  Let $z^i_j = \chi^{K_i + A_j^i} $
be the characteristic vector of the corresponding 2-edge-connected
subgraph of $G$. Then, we have
\begin{eqnarray*}
z^i = \sum_{j=1}^8\frac{1}{8} z^i_j = \sum_{j=1}^8 \frac{1}{8}
\chi^{K_i + A_j^i} = \chi^{K_i} + \sum_{j=1}^8 \frac{1}{8} \chi^{A_j^i} =
\chi^{K_i} + \frac{5}{8} \chi^{L_i}.
\end{eqnarray*}
Notice that $\sum_{i=1}^k \lambda_i \chi^{L_i} \leq
\frac{1}{2}\chi^{\cal{C}}$.  Next, we claim that the vector $z$ can be
written as a convex combination of 2-edge-connected subgraphs of $G$.
\begin{eqnarray*}
z  & = &  \sum_{k=1}^k \lambda_i z^i = \sum_{i=1}^k \lambda_i (\chi^{K_i} +
\frac{5}{8} \chi^{L_i}) = \sum_{i=1}^k \lambda_i \chi^{K_i} +
\frac{5}{8} \sum_{i=1}^k \lambda_i \chi^{L_i}\\  & \leq &
(\chi^E - \frac{1}{2} \chi^{\cal{C}}) + \frac{5}{16}\chi^{\cal{C}} = 
\chi^E - \frac{3}{16}\chi^{\cal{C}}.
\end{eqnarray*}
This concludes the proof of Lemma \ref{1-13/16}.
\end{proof}

\subsection{Coloring algorithms: Proofs of Lemmas \ref{lem:newlemma} and \ref{5/8}}\label{subsec:main-coloring-algs}

\topdownFewLM*

\begin{proof}
Let $\p= \chi^{\lm(H,T)}+3\chi^{L}$. We show that there is
$T$-admissible top-down $(\p,5)$-coloring algorithm of the links in
$L$.

Let $\gamma$ be the $(\p,5)$-coloring of $L$ that we
maintain. Initially, we have $\gamma(\ell)=\emptyset$ for $\ell \in
L$. Suppose we want to color link $\ell=uv$ at some iteration of the
algorithm.

\paragraph{Coloring Rule:} Give $\ell$ the colors in $\bar{\gamma}_2(\Lu)$ if $u$ is a leaf in $T$. If $u$ is not a leaf, give $\bar{\gamma}_1(\Lu)$ to $\ell$. Similarly, give $\ell$ the colors in $\bar{\gamma}_2(\Rv)$ if $v$ is a leaf in $T$, and if $v$ is not a leaf, give $\bar{\gamma}_1(\Rv)$ to $\ell$. If $\ell$ has fewer than three distinct colors,  give it arbitrary colors that it does not already have.

 	\vspace{5pt}

We now prove that this top-down coloring algorithm is
	$T$-admissible. Consider an $e \in T$. If $e$ is an internal edge of $T$
	(not incident on any leaf), then since the graph is essentially
	4-edge-connected we have $|\cov(e)|\geq 3$. Let
	$\ell_1,\ell_2,\ell_3$ be three of the links in $\cov(e)$ with the
	highest LCAs. When coloring $\ell_1$, edge $e$ receives three new
	colors by Observation \ref{firstlink}. Now consider
	the iteration in which the algorithm colors $\ell_i$ for some $i \in \{2,3\}$. At
	the time of coloring $\ell_i$,  the top-down coloring algorithm that we
	described above will give $\ell_i$ at least one color that an ancestor
	of $e$ is missing since $e$ is either in $\mathcal{R}_{\ell_i}$ or
	$\mathcal{L}_{\ell_i}$. By Observation \ref{ansectory}, we can conclude
	that $e$ receives a new color after coloring $\ell_i$.  Thus, after we
	have colored link $\ell_3$, edge $e$ has received at least $3+1+1=5$ colors.
	
	If $e$ is incident to a leaf, then $|\cov(e)| \geq 2$.  Let
        $\ell_1, \ell_2$ be two of the links in $\cov(e)$ with the
        highest LCAs.  When coloring $\ell_1$, edge $e$ receives three
        new colors by Observation \ref{firstlink}.  When coloring
        $\ell_2$, receives two colors that $e$ is missing by
        Observation \ref{ansectory}. So in total $e$ receives at least
        $3+2 = 5$ colors.  This concludes the proof of
        $T$-admissibility of the coloring algorithm.
	
Finally notice that each link in $L$ receives at least three colors by
construction. Moreover, if $\ell \in \lm(H,T)$, then $\ell$ is colored
with at most four colors. Therefore, the number of colors given to a link
$\ell$ is at most $\pel$ as desired.\end{proof}

To prove Lemma \ref{5/8}, we need a different strategy to handle the
leaf-matching links.  In fact, there is only one case in which
coloring a leaf-matching link is problematic, which we describe next.
Recall that the top-down coloring algorithm colors the links in any
order that respects the partial order according to their LCAs.

\begin{definition}
Consider $\ell \in \lm(H,T)$ where $\ell = uv$.  Let $\ell_u$ be 
the (only) other link that is incident on $u$ and $\ell_v$ be the (only)
other link incident on $v$.  If $\ell$ is colored after both $\ell_u$
and $\ell_v$, then we say that link $\ell$ is a \em{bad link}.
\end{definition}

For example, suppose vertex $u$ and $v$ each have degree three in $H$.
If the LCA of $\ell$ is lower than that of either $\ell_u$ or
$\ell_v$, then $\ell$ is a bad link.  We call such links ``bad'' for
the following reason.  Suppose For $p,q \in \Zplus$, suppose that we
have a partial $ (p\cdot \chi^L,q)$-coloring $\gamma$ of $L$ obtained
during some iteration of a top-down coloring algorithm.  Suppose that
$u$ and $v$ each have degree three, and suppose that both $\ell_u$ and
$\ell_v$ are both colored in $\gamma$, but link $\ell$ has not yet
been colored.  Before we color link $\ell$, the leaf edges $e_u$ and
$e_v$ (a {\em leaf edge} is the unique edge in $T$ incident to a leaf)
adjacent to $\ell_u$ and $\ell_v$, respectively, are each missing
$q-p$ colors.  If these two sets of missing colors are disjoint and $p
< 2(q-p)$,~\footnote{If $p \geq 2(q-p)$, then $p/q \geq 2/3$, which is
  not small enough for our applications.} then we will not be able to
color the link $\ell$ with $p$ colors so that $\ell_u$ and $\ell_v$
receive all $q$ colors.

To address this issue, consider the case in which $\ell$ is a
leaf-matching ink and our algorithm colors the links $\ell_u, \ell_v,
\ell$ in this order.  When we color $\ell_v$, we want the respective
set of $p$ colors to sufficiently overlap with the set of $p$ colors
already assigned to $\ell_u$; in other words, we want the set of
colors missed by $e_u$ and $e_v$ to overlap.  This way, we will be
able to ensure that $e_u$ and $e_v$ receive all $q$ colors when we
finally color the link $\ell$ with $p$ colors.  If all leaf-matching
links are vertex-disjoint, then notice that $\ell_u$ and $\ell_v$ are
not leaf-matching links.  Furthermore, link $\ell_u$ will share an
endpoint with at most one leaf-matching link, which in this case is $\ell$.  If link $uv$ is a
leaf-matching link, then we say $u$ and $v$ are {\em leaf-mates}.
This is the intuition behind the proof of Lemma \ref{5/8}, which we
now present.

\topdownfiveovereight*

\begin{proof}
We introduce a top-down $(5\chi^L,8)$-coloring algorithm of $L$, and
we then prove that it is $T$-admissible.

Since this is a top-down coloring algorithm, we sort the links by the
height of their LCA.  When we color a link $\ell$, we give it five
different colors before moving to the next link. Hence, the algorithm
runs in $|L|$ iterations. After each iteration $i \in
\{1,\ldots,|L|\}$ of the algorithm, we have a partial $(5\chi^L,8)$ coloring of $L$, namely $\gamma^i$.

 We show that our coloring algorithm will maintain two additional invariants:
	\begin{itemize}
	\item[(a)] For any coloring $\gamma^i$, an edge $e$ can only
          miss $8,3,1,$ or $0$ colors for $i \in \{1,\ldots,|L|\}$.
	\item[(b)] If both $u$ and $v$ have degree three in $H$, and
          if $\ell=uv$ is a leaf-matching link for $T$, then in any
          coloring $\gamma^i$ for which both $e_u$ and $e_v$ are
          missing a color, they miss a common color in
          $\gamma^i$. (For leaves $u$ and $v$ in $T$, let $e_u$ and
          $e_v$ be the leaf edges in $T$ incident on $u$ and $v$,
          respectively.)
	\end{itemize}
	Suppose we are performing iteration $i+1$ of the algorithm and we want
        to color link $\ell = uv$.

\paragraph{Coloring Rules:}  Depending on $u$ and $v$ we will
do one of the following.  We classify the root as an internal vertex.

\begin{itemize}
	\item[Case 1.] If both $u$ and $v$ are internal vertices in $T$,
          then give $\ell$ all colors in $\bar{\gamma}^i_2(\Lu) \cup
          \bar{\gamma}^i_2(\Rv)$.  At this point $\ell$ will have at
          most four colors. Give a color that $\ell$ does not already
          have until it has five distinct colors.

	\item[Case 2.] {If $u$ is a leaf in $T$ and $v$ is an internal
          vertex of $T$, then we consider two cases.
			
\begin{enumerate}

\item[Case 2a:] If $\ell_{uw}$ is a bad link and link $\ell_w$ is
  already colored (where $\ell_{uw}$ is the link between $u$ and
  $w$, and $\ell_w$ be the other link incident on $w$), then we choose
  five colors $C'$ for $\ell$ in the following way.  Let 
 $\gamma^i(\ell_w)$ is its set of five colors already assigned to $\ell_w$.
  By
  Claim \ref{coloring-shenanigan} we can choose five colors $C'$ for
  $\ell$ such that $\bar{\gamma}^i_1(\Lu)\in C'$,
  $\bar{\gamma}^i_1(\Rv)\in C'$, $|C'\cap \bar{\gamma}^i_3(\Lu)|\geq
  2$, $|C'\cap \bar{\gamma}^i_3(\Rv)|\geq 2$, and $|C'\cap
  \gamma^i(\ell_w)|\geq 3$. (Specifically, let $a =
  \bar{\gamma}^i_1(\Lu), b = \bar{\gamma}^i_1(\Rv), A =
  \bar{\gamma}^i_3(\Lu), B = \bar{\gamma}^i_3(\Rv), C_5 =
  \gamma^i(\ell_w)$ and $S = C'$.)

\item[Case 2b:] Otherwise (i.e., if $\ell_{uw}$ is a bad link and
  $\ell_w$ is not already colored, or if $\ell_{uw}$ is not a bad
  link), give $\ell$ the colors from $\bar{\gamma}^i_2(\Rv)$ and all
  colors in $\bar{\gamma}^i_3(\Lu)$.  If $\ell$ has fewer than five
  distinct colors, we give it any color it does not already have until
  it has five distinct colors.
				
\end{enumerate}
}
		
\item[Case 3.] {If both $u$ and $v$ are leaves in $T$, then we consider
  two cases.  Let $e_u$ and $e_v$ be the edges in the tree incident on $u$ and
 $v$, respectively. 

\begin{enumerate}

\item[Case 3a:] If both $u$ and $v$ have degree three in $H$, then by
  invariant (b), if 
$e_u$ and $e_v$ are each missing at least one color,
then there is a color $c$ that both
 $e_u$ and $e_v$ are missing.  We
  first give color $c$ to $\ell$.  Then we give colors
  $\bar{\gamma}^i_3(\Lu)\setminus\{c\}$ and
  $\bar{\gamma}^i_3(\Rv)\setminus \{c\}$ to $\ell$.

\item[Case 3b:] If at least one vertex has degree greater than three in $H$
  (say $v$),
then give colors
  $\bar{\gamma}^i_3(\Lu)$ and $\bar{\gamma}^i_2(\Rv)$ to $\ell$.

\end{enumerate}

}

\end{itemize}

\begin{claim}
The above top-down coloring algorithm preserves invariant (a).
\end{claim}
	
\begin{cproof}
We proceed by induction on the iteration of the above top-down
coloring algorithm.  It is easy to see for $\gamma^0$ the invariant
holds.  So we assume the invariant holds before the iteration $i$ in
which we color link $\ell=uv$. Consider an edge $e\in P_\ell$, and
assume without loss of generality $e\in \Rv$. By the induction
hypothesis, $e$ is missing 8, 3, 1 or 0 colors before coloring
$\ell$. If $e$ is missing eight colors, all the colors we give to
$\ell$ are new for $e$, hence after coloring $\ell$, $e$ will miss
three colors. So suppose $e$ is missing three colors before we color
link $\ell$.  But notice in all coloring rules $\ell$ will be colored
with at least two colors from $\bar{\gamma}^{i}_3(\Rv)$.  This means
that after coloring $\ell$, edge $e$ will miss at most one color. So
invariant (a) holds after coloring $\ell$.
	\end{cproof}
	
	Next, we show that invariant (b) also holds after coloring $\ell$. 
	
	\begin{claim}
		The above top-down coloring algorithm preserves invariant (b).
	\end{claim}
	
	\begin{cproof}
		Again, we proceed by induction.  We assume the
                invariant holds before the iteration in which we color
                link $\ell=uv$.  If neither $u$ nor $v$ have
                leaf-mates, 
then the invariant holds after coloring
                link $\ell$.  Thus, either (i) $\ell$ is leaf-matching
                or (ii) without loss of generality, $u$ is a leaf and
                has a leaf-mate and $v$ is an internal vertex.
		
		Suppose $\ell$ is a leaf-matching link for $T$.  If
                either $u$ or $v$ have degree greater than three in
                $H$, then the invariant holds after we color $\ell$.
                So assume both $u$ and $v$ have degree three in $H$.
                Let $e_u$ and $e_v$ be the leaf edges incident on $u$
                and $v$, respectively. Also let $\ell_u$ and $\ell_v$
                be the other links incident on $u$ and $v$,
                respectively. Since leaf-matching links for $H$ are
                disjoint, neither $\ell_u$ nor $\ell_v$ is
                leaf-matching.  If $\ell$ is not a bad link, then
                $\ell$ is colored before either $\ell_u$ or $\ell_v$.
                Before we color $\ell$, either $e_u$ or $e_v$ is
                missing eight colors.  After we color $\ell$, either
                $e_u$ and $e_v$ are missing the same three colors, or
                one is missing three colors and the other is missing
                zero colors.  Otherwise, $\ell$ is a bad link.  Now,
                consider the case in which $\ell$ is colored after
                both $\ell_u$ and $\ell_v$ have already been colored.
                Since both $e_u$ and $e_v$ are missing a common color,
                after coloring $\ell$, $e_u$ and $e_v$ are each
                missing zero colors.
		
		Now consider the case in which $u$ is a leaf in $T$
                and $v$ is an internal vertex of $T$.  Suppose $u$ has
                leaf-mate $w$ adjacent to link $\ell_w$ (which is not
                a leaf-matching link).  Moreover, we can assume that
                both $u$ and $w$ have degree three in $H$.
If $\ell_w$ is to be colored
                after $\ell$, then $e_w$ is missing eight colors both
                before and after coloring $\ell$.  Therefore, clearly
                there is a color that both $e_u$ and $e_w$ are missing
                after coloring $\ell$.  Now, consider the remaining
                case: assume that $\ell_w$ was colored before $\ell$
                in the partial coloring.  Then when coloring $\ell$
                the coloring rule is that of Case 2a.  This rule
                ensures that the set of colors we give to $\ell$ has
                three common elements with the set of colors we gave
                to $\ell_w$.  After coloring $\ell$, the set of the
                colors that $e_u$ and $e_w$ received are exactly the
                colors in $\ell$ and $\ell_w$, respectively.  In
                addition $e_u$ and $e_w$ each miss exactly three colors in
                this partial coloring.  Therefore, the set of colors
                $e_u$ is missing is not disjoint from the colors that
                $e_w$ is missing, and both $e_u$ and $e_w$ are missing
                a common color.
	\end{cproof}

	\begin{claim}
		The above top-down coloring algorithm is
		$T$-admissible.
	\end{claim}
	
	\begin{cproof}
		We now prove admissibility. Let $e$ be an edge in
                $T$. First assume $|\cov(e)| \geq 3$. So there are at
                least three links $\ell_1$, $\ell_2$, and $\ell_3$ in
                $\cov(e)$ labeled by their LCA ordering. When the
                algorithm colors $\ell_1$ since edge $e$ is missing
                all eight colors before coloring $\ell_1$ and all the five
                colors we use for $\ell_1$ are distinct, edge $e$
                receives five new colors by Observation
                \ref{firstlink}. Later, the algorithm colors $\ell_2$
                and $e$ receives at least two more new colors. This is
                because of the following: in every case of the
                coloring rules, ancestors of edge $e$ receive at least
                two new colors. By Observation \ref{ansectory} both these colors are new for $e$.  With a similar argument, when $\ell_3$ is colored, if $e$ is still missing a color, it receives its final missing color.
		
		If on the other hand we have $|\cov(e)|=2$, edge $e$
                is a leaf edge.  Let $\ell_1$ and $\ell_2$ be the two
                links that are covering $e$ labeled by the LCA
                ordering. When $\ell_1$ is colored, $e$ receives five
                new colors since all colors are new for $e$. At the
                iteration that we color $\ell_2$, the algorithm either
                applies a rule in Case 2 or in Case 3. In both cases,
                three different missing colors from ancestors of $e$
                are given to $\ell_2$. Hence, by Observation
                \ref{ansectory} edge $e$ receives the three missing
                colors.
	\end{cproof}

	In order to finish the proof we just need to prove the following claim.
	\begin{claim}\label{coloring-shenanigan}
		Let $C$ denote a set of eight distinct colors.  Let $a,b \in C$ and let $A,
		B, C_5 \subset C$ such that $a \in A, b \in B$ and $|A| = |B| =
		3$ and $|C_5| = 5$.  Then we can find $S \subset C$ such that $|S| =
		5$ and
		\begin{enumerate}
			\item  $a \in S$ and $b \in S$,
			\item $|S \cap A| \geq 2$, 
			\item $|S \cap B| \geq 2$, and 
			\item $|S \cap C_5| \geq 3$.
		\end{enumerate}
	\end{claim}
	\begin{cproof}
		If $|A \cap B| = 0$, then observe that $|(A \cup B) \cap C_5| \geq
		3$.  If $|(A \cup B) \cap C_5| = 3$, then set $S = (A \cup
		B) \setminus{c}$ where $c \neq a,b$ and $c \notin C_5$.  If $|(A \cup
		B) \cap C_5| \geq 4$, then set $S = (A \cup B) \setminus{c}$ where
		$c \neq a,b$.
		
		If $|A \cap B| = 1$, then if $|(A \cup B) \cap C_5| \geq 3$, let
		$S = A \cup B$.  So assume $|(A \cup B) \cap C_5| = 2$.  Then
		$A\cup B$ contains a color $c$ such that $c \neq a,b$ and $c \notin
		C_5$.  Let $S = (A \cup B) \setminus{c}$ and add an arbitrary new
		color from $C_5$ to $S$.

		If $|A \cap B| = 2$, then if $|(A \cup B) \cap C_5| \geq 2$, let
		$S = A \cup B$ and add an arbitrary new color from $C_5$.  If 
		$|(A \cup B) \cap C_5| = 1$, then there is some color $c \in A \cup B$
		such that $c \neq a,b$ and $c \notin C_5$.  Let $S = (A\cup
		B) \setminus{c}$ and add two new colors from $C_5$ to $S$.
		
		If $|A \cap B| = 3$, then let $c_1, c_2$ and $c_3$ be any three
		colors in $C_5 \setminus{\{a,b\}}$.  Set $S = \{a,b,c_1,c_2,c_3\}$.~\end{cproof}
	This concludes the proof.
\end{proof}

\subsection{1-trees with few leaf-matching links: Proof of
  Lemma \ref{rainbow2}}\label{sec:lmvd}

\rainbowDecom*

\begin{proof}
Let $C_{odd}$ denote the collection of odd cycles in $\C$.  For each
cycle $C$ in $\C$, consider an arbitrary partition of the edges into
adjacent pairs, leaving at most one edge $e_C$ unpaired if $C$ has odd
length.  Since the number of odd length cycles is even (because $G$ is
a cubic graph and has an even number of vertices), we can arbitrarily
pair the edges in the set $\{e_C\}_{C \in C_{odd}}$.  By Observation
\ref{obs:in-subtour}, we have $y = \chi^E - \frac{1}{2}\chi^{\C} \in
\HK(G)$ and $y(\delta(v)) = 2$ for all $v \in V$.  Thus, if we apply
Theorem \ref{boydsebo-rainbow}, we find a set of 1-trees
$\{K_1,\ldots,K_k\}$ such that each $1$-tree uses exactly one edge
from each pair.  Notice that each edge in $E \setminus{\C}$ belongs to
$K_i$ for $i \in \{1,\ldots,k\}$.

Observe that for any such partition of edges in $\C$, which pairs at
most one edge $e_C$ from an odd cycle $C$ with an edge in another odd
cycle, we have $\lm(G,K_i)\subseteq \{e_C\}_{C \in \C_{odd}}$.  Since
the set of edges $\{e_C\}_{C \in \C_{odd}}$ are vertex disjoint, any
such partition of $C$ results via Theorem \ref{boydsebo-rainbow} in a
set of 1-trees that satisfy property (ii).  
Moreover, notice that if
$C \in C_{odd}$ is a triangle, then $e_C$ will never be a
leaf-matching link for any $K_i$.\footnote{We use Lemma \ref{rainbow2} to
  prove Lemma \ref{1-13/16} which assumes that $G$ is essentially
  4-edge-connected.  A cycle cover of such a graph cannot contain a
  triangle.  However, we choose to state Lemma \ref{rainbow2} using
  the fewest possible assumptions.}

To show property (i), for each odd cycle in $C_{odd}$ with length at
least five, we choose five edges from this odd cycle and label them
$e^j_C$ for each $j \in \{1,\ldots, 5\}$.  For a triangle $C$ in
$C_{odd}$, we choose a single edge to be $e_C$ and set $e^j_C$ to be
equal to this edge for all $j \in \{1,\ldots, 5\}$.  Next we construct
five partitions of edges in $\C$: for $j \in \{1,\ldots,5\}$ only
edges in $\{e^j_C\}_{C \in \C_{odd}}$ are paired with an edge in
another cycle.  Then we apply Theorem \ref{boydsebo-rainbow} fives
times, one for each of the five partitions.  Let $\{K_1^j, \dots,
K_{k_j}^j\}$ denote the 1-trees in the convex combination obtained for
the $j^{th}$ partition.  Note that if cycle $C$ has length at least
five, then $e^j_C$ is leaf-matching at most
half the time in this convex combination.  Thus, the union of the
1-trees constructed for these five partitions satisfy properties (i)
and (ii).
\end{proof}

We remark that if the cycle cover $\C$ of $G$ contains only even
cycles (e.g., when $G$ is 3-edge-colorable), then the $1$-trees found
via Lemma \ref{rainbow2} have no leaf-matching links.  For such
graphs, we can write $\chi^E - \frac{2}{15} \chi^{\C}$ as a convex
combination of 2-edge-connected subgraphs.  This yields the following
theorem, whose complete proof can be found in Appendix
\ref{app:3edgecolorable}.

\begin{restatable}{theorem}{edgeColorable}\label{thm:3-edge-colorable}
	Let $G=(V,E)$ be a 3-edge-connected, cubic, 3-edge-colorable graph. Then
	$\frac{13}{15} \chi^E$ can be written as a convex combination of
	2-edge-connected subgraphs of $G$. 
\end{restatable}

\subsection{Improved bounds for multigraphs}

The ideas used to prove Theorem \ref{uniformcover-cubic} can be
combined with the fact that 3-edge-connected cubic graphs have a cycle
cover covering all 3-edge cuts and 4-edge cuts \cite{bit13} to improve
the factor of $\frac{7}{8}$ when we are allowed to double edges.

\begin{theorem}\label{doubled-edges}
	Let $G=(V,E)$ be a 3-edge-connected cubic graph. The vector
	$\frac{41}{47} \chi^E$ can be written as a convex combination of
	2-edge-connected multigraphs of $G$.
\end{theorem}

This improves over the factor of $\frac{15}{17}$ in \cite{uniform}. We
remark that Theorem \ref{doubled-edges} implies that for any cost
vector on a 3-edge-connected cubic graph that is 
optimized by the vector $\frac{2}{3}\chi^E$, there is a $1.309$-approximation
algorithm for 2ECS. Such cost vectors include unit costs on
3-edge-connected cubic graphs for which a $1.2$-approximation
algorithm is known for 2ECS~\cite{boyd2014frac}, and node-weighted costs on
3-edge-connected cubic graph for which a $1.3$-approximation algorithm
is known for 2ECS \cite{uniform}.

To prove Theorem \ref{doubled-edges}, we first apply Theorem \ref{bit}
to obtain a cycle cover $\C$ of $G$ that covers all 3- and 4-edge cuts
of $G$.  

\begin{lemma}\label{y2}
Let $G=(V,E)$ be a 3-edge-connected cubic graph and $\mathcal{C}$ be a
cycle cover of $G$ that covers 3-edge cuts and 4-edge cuts of
$G$. Then $\frac{3}{5} \chi^E + \frac{2}{5} \chi^{\C}$ can be written
as a convex combination of 2-edge-connected multigraphs of $G$.
\end{lemma}
\begin{proof}
	Consider graph $H=G/\mathcal{C}$. Notice that $H$ is
        5-edge-connected, which means that the vector
        $\frac{2}{5}\chi^{E(H)}$ for $H$ is in $\HK(H)$. The
        polyhedral proof of Christofides algorithm implies that the
        vector $\frac{3}{5}\chi^{E(H)}$ can be written as a convex combination of 2-edge-connected multigraphs of $H$, namely $F_1,\ldots,F_k$. Notice that for $i\in \{1,\ldots,k\}$, the set of edges $\C+F_i$ induces a 2-edge-connected multigraph on $G$.
\end{proof}
\begin{proof}[Proof of Theorem \ref{doubled-edges}]
	From Lemma \ref{ccpoint-cubic}, we have that $\chi^E -
        \frac{3}{16}\chi^{\C}$ can be written as convex combination of
        2-edge-connected subgraphs of $G$.  Combining this with Lemma
        \ref{y2}, we have
\begin{eqnarray}
\frac{32}{47}(\chi^E - \frac{3}{16}\chi^{\C}) +
\frac{15}{47}(\frac{3}{5} \chi^E + \frac{2}{5} \chi^{\C}) & = & \frac{41}{47}\chi^E.
\end{eqnarray}
\end{proof}

\section{2ECM for half-integer square points}\label{sec:squarepoints}
In this section, our goal is to prove the following theorem.  
\squarepoint*

We will use the following theorem due to Boyd and
Seb{\H{o}}~\cite{boydsebo-mp}.

\begin{theorem}[\cite{boydsebo-mp}]
Let $x$ be a half-integer square point. The graph $G_x$
has a Hamiltonian cycle that contains all the 1-edges of $x$ and
opposite half-edges from each half-square in $G_x$. Moreover, this
Hamiltonian cycle can be found in time polynomial in the size of
$G_x$.
\end{theorem}
Let $H$ be such a Hamiltonian cycle of $G_x$. For simplicity, let $A$
be the set of 1-edges of $G_x$, $B$ be the set of half-edges of $G_x$
that are in $H$, and $C$ be the half-edges of $G_x$ that are not in
$H$. Thus, the incidence vector of $H$ is
\begin{equation*}
\chi^H_e = \left\{ \begin{array}{ll}
1 & \mbox{if $e  \in A$};\\
1 & \mbox{if $e\in B$};\\
0 & \mbox{if $e\in C$}.
\end{array} \right. 
\end{equation*}
In order to use $H$ as part of a convex combination in proving Theorem
\ref{square-point}, we need to be able to save on edges in $B$. To
this end, we introduce the following definitions.
\begin{definition}\label{def:r-vec}
	For $\alpha>0$, let $r^{\alpha,x}$ to be the vector in $\mathbb{R}^{E(G_x)}$ where 

\begin{equation*}	r^{\alpha,x}_e = \left\{ \begin{array}{ll}
		1+\alpha & \mbox{if $e \in A$};\\
		\frac{1}{2} & \mbox{if $e\in B$};\\
		1-\alpha	  & \mbox{if $e\in C$}.
	\end{array} \right. 
	\end{equation*}
\end{definition}

\begin{definition}\label{property}
Let $G=(V,E)$ be a graph.
We say property {\em ${P}(G,\alpha)$ holds} if the vector $\alpha \chi^E$
can be written as a convex combination
of matchings $\{M_1,\ldots,M_k\}$ of $G$ such that $G_1' = (V, E\setminus
M_1), \ldots, G_k' = (V, E\setminus M_k)$ are 2-vertex-connected
subgraphs of $G$.
\end{definition}

Let $G_x$ be the support graph of a half-integer square point, and let
$G=(V,E)$ be the 4-regular 4-edge-connected graph obtained from $G_x$
by replacing each path of 1-edges with a single 1-edge and contracting
all of its half-squares.\footnote{Observe that $G$ is Eulerian and is
  therefore 4-edge-connected since the corresponding Boyd-Carr point
  is 3-edge-connected.}
\begin{restatable}{lemma}{construction}
	\label{construction}
	If $P(G,\alpha)$ holds for the graph $G$ obtained from $G_x$,
        then the vector $r^{\alpha,x}$ can be written as a convex
        combination of 2-edge-connected multigraphs of $G_x$.
\end{restatable}

Lemma \ref{construction} will be proved in Section \ref{sec:const}.
It is clear that $P(G,0)$ holds. By Lemma \ref{construction}, the
vector $r^{0,x}$ dominates a convex combination of 2-edge-connected
multigraphs of $G_x$. Hence any convex combination of vectors
$r^{0,x}$ and $\chi^H$ also dominates a convex combination of
2-edge-connected multigraphs. Thus, $\frac{2}{3}r^{0,x} +
\frac{1}{3}\chi^{H}$ dominates a convex combination of
2-edge-connected multigraphs of $G_x$. We have $\frac{2}{3}r^{0,x} +
\frac{1}{3}\chi^{H}\leq \frac{4}{3}x$. To go beyond $\frac{4}{3}$, we
need to use the half-edges less and therefore we need to account for
this by sometimes doubling 1-edges.  The property ${P}(G,\alpha)$ will
allow us to double all the 1-edges in $G_x$ that belong to a
particular matching in $G$ (i.e., an $\alpha$-fraction of the
1-edges).  In this section, our main goal is to prove the following
lemma.
\begin{theorem}\label{1/10}
For any 4-regular, 4-edge-connected graph $G$, $P(G,\frac{1}{10})$ holds.
\end{theorem}
By Lemma \ref{construction}, we have the following corollary.
\begin{corollary}\label{vector-r}
	For a half-integer square point $x$, the vector
        $r^{\frac{1}{10},x}$ dominates a convex combination of
        2-edge-connected multigraphs of $G_x$ and this convex
        combination can be found in time polynomial in the size of
        $G_x$.
	\end{corollary}

From Corollary \ref{vector-r}, the proof of Theorem \ref{square-point}
is easy. Obviously any convex combination of $r^{\frac{1}{10},x}$ and
$\chi^H$ also dominates a convex combination of 2-edge-connected
multigraphs of $G_x$. Consider the combination
$\frac{5}{7}r^{\frac{1}{10},x}+ \frac{2}{7}\chi^H$. It is easy to see
this convex combination is dominated by $\frac{9}{7}x$.

It remains to prove Lemma \ref{construction} and Theorem \ref{1/10}.
We will prove Lemma \ref{construction} in Section \ref{sec:const},
where we describe how to construct the convex combination.  Regarding
Theorem \ref{1/10}, note that $P(G,\frac{1}{10})$ is equivalent to
saying that the vector $\frac{9}{10}\chi^E$ can be written as a convex
combination of 2-vertex-connected subgraphs of minimum degree three.
This equivalent statement will be proved using Lemma \ref{9/10}.
\begin{restatable}{lemma}{nineoverten}
	\label{9/10}
	Let $G$ be a 4-regular 4-edge-connected graph. Let $T$ be a
        spanning tree of $G$ such that $T$ does not have any
        vertex of degree four. The vector $y\in \mathbb{R}^G$, where
        $y_e = \frac{4}{5}$ for $e\notin T$ and $y_e = 0$ for $e\in
        T$, dominates a convex combination of edge sets
        $\{F_1,\ldots,F_k\}$ such that $T+F_i$ is a 2-vertex-connected
        subgraph of $G$ where each vertex has degree at least three in
        $T+F_i$ for $i \in \{1,\ldots,k\}$. 
\end{restatable}

The proof of Lemma \ref{9/10} can be found in Section \ref{sec:9/10}.
In order to prove this lemma, we need a way to reduce vertex
connectivity to edge-connectivity, which is done in Section
\ref{vertex-cut}.  The main tool in the proof of Lemma \ref{9/10} is a
top-down coloring algorithm, which is detailed in Section
\ref{app:coloring}.  From Lemma \ref{9/10}, one can easily prove
Theorem \ref{1/10}.
\begin{proof}[Proof of Theorem \ref{1/10}]

Consider a half-integer square point $x$. Let $G=(V,E)$ be the graph
obtained from $G_x$ by replacing each path of 1-edges with a single
1-edge and 
contracting all the half-squares in $G_x$. Graph $G$ is
4-regular and 4-edge-connected, hence $G$ has two edge-disjoint
spanning trees $T_1$ and $T_2$ \cite{nw}. Notice that $T_1$ and $T_2$
cannot have any vertex of degree four, since for all vertices $v\in
V$, we have $\delta_{T_1}(v)\geq 1$ and $\delta_{T_2}(v)\geq 1$ while
$\delta_{T_1}(v)+\delta_{T_2}(v)\leq 4$. Hence, by Lemma \ref{9/10} we
can write vector $z^i\in \mathbb{R}^G$, with $z^i_e = 1$ for $e\in
T_i$, and $z^i_e=\frac{4}{5}$ for $e\notin T_i$ as a convex
combination of 2-vertex-connected subgraphs of $G$ where every vertex
has degree at least three, for $i\in \{1,2\}$. Now consider
$\frac{1}{2}\cdot z^1 + \frac{1}{2}\cdot z^2$: it dominates a convex
combination of 2-vertex-connected subgraphs of $G$ where every vertex
has degree at least three. Also, $\frac{1}{2}\cdot z^1 +
\frac{1}{2}\cdot z^2$ is the vector $\frac{9}{10}\chi^E$.  This
concludes the proof, since the complement of the solutions in the
convex combination form the desired convex combination of matchings.
\end{proof}

\subsection{From matching to 2ECM: Proof of Lemma \ref{construction}}\label{sec:const}
Recall that $G_x$ is the support graph of a half-integer square point
$x$, and $G=(V,E)$ is the 4-regular 4-edge-connected graph obtained
from $G_x$ by replacing each path of 1-edges with a single 1-edge and
contracting all of its half-squares.  The definition of vector
$r^{\alpha,x}$ can be found in Definition \ref{def:r-vec}, and the
definition of edge sets $A,B$ and $C$ can be found directly before.

\construction*
\begin{proof}
Since $P(G,\alpha)$ holds, we can find $\lambda_1,\ldots,\lambda_k \in
\mathbb{R}_{\geq 0}$ where $\sum_{i=1}^{k}\lambda _i = 1$,
$\alpha = \sum_{i=1}^{k}\lambda_i \chi^{M_i}$, and $M_i$ is a
matching in $G$ such that graph $G_i' = (V, E\setminus M_i)$ is
2-vertex-connected for $i \in \{1,\ldots,k\}$.  Specifically, for each
$i\in \{1,\ldots,k\}$, we create two 2-edge-connected multigraphs
$F^i_1$ and $F^i_2$, as follows.  Notice that each edge in $M_i$
corresponds to a 1-edge (an edge in $A$) in $G_x$.  For each $e\in
M_i$ we add two copies of the 1-edge corresponding to $e$ in $G_x$ to
$F^i_1$ and $F^i_2$. For each $e\notin M_i$ we add one copy of the
1-edge corresponding to $e$ in $G_x$ to $F^i_1$ and $F^i_2$.
Additionally, we assign an arbitrary orientation to each edge $e\in
M_i$. For each edge $e\in M_i$, there are two squares $Q_1$ and $Q_2$
incident on $e$. We say $e\rightarrow Q_1$ and $e\leftarrow Q_2$ if
$e$ is oriented from the endpoint in $Q_2$ towards the endpoint in
$Q_1$.

Consider a half-square $Q$ with vertices $u_1,u_2,u_3$ and $u_4$ in
$G_x$. There are four 1-edges incident on $Q$, namely $f_j$ for $j \in
\{1,\ldots,4\}$, where $f_j$ is incident to $u_j$. Since $M_i$ is a
matching in $G$, at most one of $\{f_1,f_2,f_3,f_4\}$ belongs to
$M_i$. If one of $\{f_1,\ldots,f_4\}$ is in $M_i$ we can assume without loss of generality that $f_1\in M_i$. If
$f_1\rightarrow Q$, then we add to $F^1_i$ the two half-edges in $Q$
that are not incident on $u_1$. If $f_1\leftarrow Q$, then we add
to $F^1_i$ the two half-edges in $Q$ that are incident to $u_1$
together with the other half-edge in $Q\cap C$. For $F^2_i$ we do the
opposite: If $f_1\leftarrow Q$, then we add to $F^2_i$ the two
half-edges in $Q$ that do not have as endpoint $u_1$, and if
$f_1\rightarrow Q$, then we add to $F^2_i$ the two half-edges in $Q$
that are not incident to $u_1$ together with the other half-edge in
$Q\cap C$.  See Figure \ref{fig:half-squares-red} for an illustration.
If none of $\{f_1,\ldots,f_4\}$ belong to $M_i$, we add both edges in
$C\cap Q$ to $F^i_1$ and $F^i_2$. We also arbitrarily choose an edge
in $Q \cap B$ to add to $F_1^i$ and add the other edge in $Q \cap B$
to $F_2^i$.

\begin{figure}[h]
	\centering
	\begin{subfigure}[b]{0.3\linewidth}
		\begin{tikzpicture}[scale=0.8]
		
		\draw [dashed] [black, line width=0.7mm] plot [smooth, tension=0] coordinates {(0,1) (1,2)};
		
		\draw [-] [black, line width=0.7mm] plot [smooth, tension=0] coordinates {(1,2) (2,1)};
		
		\draw [dashed] [black, line width=0.7mm] plot [smooth, tension=0] coordinates {(2,1) (1,0)};
		
		\draw [-] [black, line width=0.1mm] plot [smooth, tension=0] coordinates {(1,0) (0,1)};

		\draw [dashed] [black, line width=0.1mm,xshift=3.8 cm] plot [smooth, tension=0] coordinates {(0,1) (1,2)};
		
		\draw [-] [black, line width=0.7mm,xshift=3.8 cm] plot [smooth, tension=0] coordinates {(1,2) (2,1)};
		
		\draw [dashed] [black, line width=0.7mm,xshift=3.8 cm] plot [smooth, tension=0] coordinates {(2,1) (1,0)};
		
		\draw [-] [black, line width=0.1mm,xshift=3.8 cm] plot [smooth, tension=0] coordinates {(1,0) (0,1)};
		
		\draw [->] [black, line width=0.3mm] plot [smooth, tension=0] coordinates {(2,1) (3.7,1)};
		
		\draw[black,fill=white] (1,0) ellipse (0.1 cm  and 0.1 cm);
		\draw[black,fill=white] (0,1) ellipse (0.1 cm  and 0.1 cm);
		\draw[black,fill=white] (2,1) ellipse (0.1 cm  and 0.1 cm);
		\draw[black,fill=white] (1,2) ellipse (0.1 cm  and 0.1 cm);
		\draw[black,fill=white] (4.8,0) ellipse (0.1 cm  and 0.1 cm);
		\draw[black,fill=white] (3.8,1) ellipse (0.1 cm  and 0.1 cm);
		\draw[black,fill=white] (5.8,1) ellipse (0.1 cm  and 0.1 cm);
		\draw[black,fill=white] (4.8,2) ellipse (0.1 cm  and 0.1 cm);
		
		\node (Q) at (4.8,1) {{\large$Q$}};
		\end{tikzpicture}
		\caption{$F^i_1$}
	\end{subfigure}
	\hspace{10mm}
	\begin{subfigure}[b]{0.3\linewidth}
		\begin{tikzpicture}[scale=0.8]
		
		\draw [dashed] [black, line width=0.8mm] plot [smooth, tension=0] coordinates {(0,1) (1,2)};
		
		\draw [-] [black, line width=0.1mm] plot [smooth, tension=0] coordinates {(1,2) (2,1)};
		
		\draw [dashed] [black, line width=0.1mm] plot [smooth, tension=0] coordinates {(2,1) (1,0)};
		
		\draw [-] [black, line width=0.8mm] plot [smooth, tension=0] coordinates {(1,0) (0,1)};

		\draw [dashed] [black, line width=0.8mm,xshift=3.8 cm] plot [smooth, tension=0] coordinates {(0,1) (1,2)};
		
		\draw [-] [black, line width=0.1mm,xshift=3.8 cm] plot [smooth, tension=0] coordinates {(1,2) (2,1)};
		
		\draw [dashed] [black, line width=0.8mm,xshift=3.8 cm] plot [smooth, tension=0] coordinates {(2,1) (1,0)};
		
		\draw [-] [black, line width=0.8mm,xshift=3.8 cm] plot [smooth, tension=0] coordinates {(1,0) (0,1)};
		
		\draw [->] [black, line width=0.3mm] plot [smooth, tension=0] coordinates {(2,1) (3.7,1)};
		
		\draw[black,fill=white] (1,0) ellipse (0.1 cm  and 0.1 cm);
		\draw[black,fill=white] (0,1) ellipse (0.1 cm  and 0.1 cm);
		\draw[black,fill=white] (2,1) ellipse (0.1 cm  and 0.1 cm);
		\draw[black,fill=white] (1,2) ellipse (0.1 cm  and 0.1 cm);
		\draw[black,fill=white] (4.8,0) ellipse (0.1 cm  and 0.1 cm);
		\draw[black,fill=white] (3.8,1) ellipse (0.1 cm  and 0.1 cm);
		\draw[black,fill=white] (5.8,1) ellipse (0.1 cm  and 0.1 cm);
		\draw[black,fill=white] (4.8,2) ellipse (0.1 cm  and 0.1 cm);
		
		\node (Q) at (4.8,1) {{\large$Q$}};

		\end{tikzpicture}
		\caption{$F^i_2$}
	\end{subfigure}
	\label{fig:5.0}
	\caption{Solid edges belong to $B$ and dashed edges  belong to
		$C$.  The directed edge belongs to the matching.
		Thick edges represent those half-edges that are added to $F_1^i$ and $F_2^i$, respectively.}  \label{fig:half-squares-red}
\end{figure}

We conclude this proof with the following two key claims.

\begin{claim}
	The graphs induced on $G_x$ by edge sets $F^i_1$ and $F^i_2$
        are 2-edge-connected multigraphs of $G_x$ for $i \in \{1,\ldots,k\}$.
\end{claim}
	\begin{cproof}
Since the construction of $F_1^i$ and $F_2^i$ are symmetric, it is
enough to show this only for $F^i_1$.  First notice that for every
vertex $v\in G_x$, we have $|F^i_1\cap \delta(v)|\geq 2$. Let $e$ be
the 1-edge incident on $v$. If $e\in M_i$, then we have two copies of
$e$ in $F^i_1$ so we are done. If $e\notin M_i$, then $F^i_1$ contains
only one copy of $e$.  However, by construction, in the half-square
that contains $v$, we will have at least one half-edge in $F^i_1$ that
is incident to $v$.

We proceed by showing that for every set of edges $D$ in $G_x$ that
forms a cut (i.e., whose removal disconnects the graph $G_x$), we have
$|D\cap F^i_1|\geq 2$. Clearly, if $D$ contains two or more 1-edges,
since $F^i_1$ contains all the 1-edges, we have $|D\cap F^i_1|\geq 2$.
So assume $|D\cap A|=1$; $D$ contains exactly one 1-edge $e$ of
$G_x$. If $e\in M_i$, we are done as the matching will take two copies
of $e$. Thus, we may assume $e\notin M_i$.  Notice that for any edge
cut $D$, $D$ contains either zero or two edges from every
half-square. Hence, we can pair up the half-edges in $D$. Let
$e_1,\ldots,e_n,f_1,\ldots,f_m$ and
$e'_1,\ldots,e'_n,f'_1,\ldots,f'_m$ be the half-edges in $D$ such that
$e_j$ and $e'_j$ belong to the same half-square and are opposite
edges, and $f_j$ and $f'_j$ belong to the same half-square and share
an endpoint.  Notice that while we can have $m=0$ or $n=0$, it must be
the case that $n+m \geq 1$, since $G_x$ is 2-edge-connected and hence $D$
must contain two edges from at least one half-square.  Note that
$D\cap F^i_1$ contains edge $e$. For a contradiction, suppose
that $|D\cap F^i_1|=1$. In this case, we must have $n=0$ since in our
construction we take at least one half-edge from every pair of
opposite half-edges. (In other words, if $n \geq 1$, then $D$ and
$F_1^i$ must have at least one half-edge in common.) For $j \in
\{1,\ldots,m\}$, let $u_j$ be the endpoint that $f_j$ and $f'_j$ share
and let $g_j$ be the 1-edge incident to $u_j$.  Notice that $D'=e\cup
\{\bigcup_{j=1}^m g_j\}$ forms a cut in $G_x$ that only contains
1-edges. Thus, $D'$ is also a cut in $G$. This implies that there is
an edge $g_j$ for some $j\in \{1,\ldots,m\}$ such that $g_j\notin
M_i$. Otherwise, $e$ is the unique edge of cut $D'$ that is not in $M_i$. This means that $G_i'=(V,E\setminus M_i)$ has a cut with only one edge, which implies that it is not 2-vertex-connected.
Since $g_j\notin M_i$, by construction $F^i_1$
contains an edge in the half-square that contains $u_j$. This implies
that $|F^i_1 \cap \{f_j,f'_j\}|\geq 1$, which is a contradiction to
the assumption that $|D \cap F^i_1| = 1$ (See Figure \ref{helpconstruction}.)

		
		\begin{figure}[h]
			\centering
			\begin{tikzpicture}[scale=0.7]
			
			\draw [-] [black, line width=0.3mm] plot [smooth, tension=0] coordinates {(-3,1.5) (-3,-0.7)};
			
			\draw [dashed] [black, line width=0.3mm] plot [smooth, tension=0] coordinates {(-1,0.5) (-2,-1.2)};
			\draw [dashed] [black, line width=0.3mm] plot [smooth, tension=0] coordinates {(-1,0.5) (0,-1.2)};

			\draw [-] [black, line width=0.3mm] plot [smooth, tension=0] coordinates {(-1,0.5) (-1,2)};
			\draw [dashed] [black, line width=0.3mm] plot [xshift=4.5cm,smooth, tension=0] coordinates {(-1,0.5) (-2,-1.2)};
			\draw [dashed] [black, line width=0.3mm] plot [xshift = 4.5cm, smooth, tension=0] coordinates {(-1,0.5) (0,-1.2)};
			\draw [-] [black, line width=0.3mm] plot [smooth, tension=0,xshift =4.5cm] coordinates {(-1,0.5) (-1,2)};

			\draw [dashed] [black, line width=0.3mm] plot [xshift=9cm,smooth, tension=0] coordinates {(-1,0.5) (-2,-1.2)};
			\draw [dashed] [black, line width=0.3mm] plot [xshift = 9cm, smooth, tension=0] coordinates {(-1,0.5) (0,-1.2)};
			\draw [-] [black, line width=0.3mm] plot [smooth, tension=0,xshift =9cm] coordinates {(-1,0.5)  (-1,2)};

			\draw [dotted] [blue, line width=0.3mm] plot [smooth, tension=0.5] coordinates {(-4,1.3) (-1,1) (8,1) (10,1.3)};
			
			\draw [dotted] [red, line width=0.3mm] plot [smooth, tension=0.5] coordinates {(-4,0.3) (-1,-0.8) (8,-0.8) (10,0.3)};
			
			\node (D') at (-4.45,1.3) {\textcolor{blue}{\large$D'$}};
			\node (D) at (-4.45,0.3) {\textcolor{red}{\large$D$}};

			\draw[black,fill=white] (-3,-0.8) ellipse (0.1 cm  and 0.1 cm);
			\draw[black,fill=white] (-3,1.5) ellipse (0.1 cm  and 0.1 cm);
			\draw[black,fill=white] (-2,-1.2) ellipse (0.1 cm  and 0.1 cm);
			\draw[black,fill=white] (-1,0.5) ellipse (0.1 cm  and 0.1 cm);
			\draw[black,fill=white] (0,-1.2) ellipse (0.1 cm  and 0.1 cm);
			\draw[black,fill=white] (-1,2) ellipse (0.1 cm  and 0.1 cm);
			\draw[black,fill=white,xshift=4.5cm] (-2,-1.2) ellipse (0.1 cm  and 0.1 cm);
			\draw[black,fill=white,xshift=4.5cm] (-1,0.5) ellipse (0.1 cm  and 0.1 cm);
			\draw[black,fill=white,xshift=4.5cm] (0,-1.2) ellipse (0.1 cm  and 0.1 cm);
			\draw[black,fill=white,xshift=4.5cm] (-1,2) ellipse (0.1 cm  and 0.1 cm);
			\draw[black,fill=white,xshift=9cm] (-2,-1.2) ellipse (0.1 cm  and 0.1 cm);
			\draw[black,fill=white,xshift=9cm] (-1,0.5) ellipse (0.1 cm  and 0.1 cm);
			\draw[black,fill=white,xshift=9cm] (0,-1.2) ellipse (0.1 cm  and 0.1 cm);
			\draw[black,fill=white,xshift=9cm] (-1,2) ellipse (0.1 cm  and 0.1 cm);

			\node (Q) at (1.25,0.3) {{\large$\ldots$}};
			
			\node (e) at (-3.3,0.5) {{\large$e$}};
			
			\node (fj) at (-1.9+4.5,-0.2) {{\large$f_j$}};
			\node (f'j) at (-0.1+4.5,-0.2) {{\large$f'_j$}};
			\node (gj) at (-1.35+4.5,1+0.5) {{\large$g_j$}};

			\node (f1) at (-1.9,-0.2) {{\large$f_1$}};
			\node (f'1) at (-0.1,-0.2) {{\large$f'_1$}};
			\node (g1) at (-1.35,1+0.5) {{\large$g_1$}};
			
			\node (fm) at (-1.9+9,-0.2) {{\large$f_m$}};
			\node (f'm) at (-0.1+9,-0.2) {{\large$f'_m$}};
			\node (gm) at (-1.4+9,1+0.5) {{\large$g_m$}};

			\node (Q) at (5.75,0.3) {{\large$\ldots$}};

			\end{tikzpicture}
			\caption{Edges in the cuts $D$ and $D'$.}  
			\label{helpconstruction}
		\end{figure}
Finally, assume that $D$ does not contain any 1-edges. In this case,
let $e_1,\ldots,e_n,f_1,\ldots,f_m$ and
$e'_1,\ldots,e'_n,f'_1,\ldots,f'_m$ be the half-edges in $D$ such that
$e_j$ and $e'_j$ belong to the same half-square and are opposite edges, and
$f_j$ and $f'_j$ belong to the same half-square and share one
endpoint. Notice that we can have $m=0$ or $n=0$ but $n+m\geq 2$, because
$D$ must contain edges from at least two half-squares (since $G_x$ is
2-vertex connected). For
$j \in \{1,\ldots,m\}$ let $u_j$ be the endpoint that $f_j$ and $f'_j$ share
and $g_j$ be the 1-edge incident on $u_j$. If $n=0$, then
$D'=\bigcup_{j=1}^{m}g_j$ forms a cut in $G$. Hence, there are two
edges $g_j$ and $g_{j'}$ such that $g_j,g_{j'}\notin M_i$. This implies
that $|F^i_1\cap \{f_j,f'_j\}|\geq 1$, and $|F^i_1\cap
\{f_{j'},f'_{j'}\}|\geq 1$. Therefore, $|D\cap F^i_1|\geq 2$. If
$n=2$, then by construction $|F^i_1\cap \{e_1,e'_1\}|\geq 1$, and
$|F^i_1\cap \{e_2,e'_2\}|\geq 1$, so we have the result. It only
remains to consider the case when $n=1$. Notice as before we have
$|F^i_1 \cap \{e_1,e'_1\}|\geq 1$. If there is $g_j$ for some $j\in
\{1,\ldots,m\}$ such that $g_j\notin M_i$, then we have $|F^i_i\cap
\{f_j,f'_j\}|\geq 1$ in which case we are done. Thus, we may assume
$g_j\in M_i$. Let $Q$ be the half-square that contains $e_1$ and $e'_1$. In
$G_i'=(V,E\setminus M_i)$ the vertex corresponding to $Q$ will be a cut
vertex, which is a contradiction.
\end{cproof}
Now we conclude the proof by proving the second and last claim.

\begin{claim}
	Let $r= \sum_{i=1}^{k} \frac{\lambda_i}{2}\chi^{F_1^i} + \sum_{i=1}^{k} \frac{\lambda_i}{2}\chi^{F_2^i}$. We have $r_e=1+\alpha$ for $e\in A$, $r_e=\frac{1}{2}$ for $e\in B$, and $r_e = 1-\alpha$ for $e\in C$, i.e. $r= r^{x,\alpha}$.
\end{claim}
\begin{cproof}
	 Let $e\in A$ (a 1-edge in $G_x$). We have $\sum_{i\in [k]:\;e\in M_i}\lambda_i = \alpha$. Therefore,
	 \begin{align*}
	 \sum_{i=1}^{k}\frac{\lambda_i}{2}\chi^{ F^i_1}_e+ \sum_{i=1}^{k}\frac{\lambda_i}{2}\chi^{ F^i_2}_e &= \sum_{i\in k:\; e\in M_i}^{k}\frac{2\lambda_i}{2}+ \sum_{i\in k:\; e\notin M_i}^{k}\frac{\lambda_i}{2}+ \sum_{i\in k:\; e\in M_i}^{k}\frac{2\lambda_i}{2}+\sum_{i\in k:\; e\notin M_i}^{k}\frac{\lambda_i}{2}\\
	 &= \alpha + \frac{1}{2}-\frac{\alpha}{2} + \alpha + \frac{1}{2}-\frac{\alpha}{2}\\
	 &=1+\alpha.
		\end{align*}
	Now consider a half-edge $e\in B$.  Let $f$ and $g$ be the 1-edges incident on the endpoints of $e$. If $f\in M_i$ and $f$ is incoming to $e$, then $e\notin F^i_1$ and $e\in F^i_2$, otherwise if $f\in M_i$ and $f$ is outgoing of $e$, then $e\in F^i_1$ and $e\notin F^i_2$. This means that if $f\in M_i$, then $\frac{\lambda_i}{2}\chi^{F^i_1}_e + \frac{\lambda_i}{2}\chi^{F^i_2}_e=\frac{\lambda_i}{2}$. Similarly, if $g\in M_i$, we have  $\frac{\lambda_i}{2}\chi^{F^i_1}_e + \frac{\lambda_i}{2}\chi^{F^i_2}_e=\frac{\lambda_i}{2}$. 	Notice that if $f\in M_i$, then $g\notin M^i$, since in $G$, edges $f$ and $g$ share an endpoint and $M_i$ is a matching.

Now, assume $f,g\notin M_i$. Let $f',g'$ be the other 1-edges incident on the square $Q$ that contains $e$. If $f'\in M_i$, then if $f'$ is incoming to $Q$, then $e\in F^i_1$ and $e\notin F^i_2$. If $f'$ is outgoing from $Q$, then $e\notin F^i_1$ and $e\in F^i_2$. In both case, $\frac{\lambda_i}{2}\chi^{F^i_1}_e + \frac{\lambda_i}{2}\chi^{F^i_2}_e=\frac{\lambda_i}{2}$. Similarly, if $g'\in M_i$. If $f,g,f',g'\notin M_i$, then exactly one of $F^i_1$ and $F^i_2$ will contain $e$. Hence, $\frac{\lambda_i}{2}\chi^{F^i_1}_e + \frac{\lambda_i}{2}\chi^{F^i_2}_e=\frac{\lambda_i}{2}$. We have,
 \begin{align*}
\sum_{i=1}^{k}\frac{\lambda_i}{2}\chi^{ F^i_1}_e+ \sum_{i=1}^{k}\frac{\lambda_i}{2}\chi^{ F^i_2}_e &= \sum_{i=1}^{k}\frac{\lambda_i}{2}= \frac{1}{2}.
\end{align*}

Now consider edge $e\in C$. Let $Q$ be the square in $G_x$ that contains $e$. Let $f,g,f',g'$ be the 1-edges incident on $Q$ such that $f,g$ are the 1-edges that are incident on the endpoints of $e$. If $f\in M_i$ and $f$ is incoming to $Q$, then $e\notin F^i_1$. Also, if $g\in M_i$ and $g$ is incoming to $Q$, then $e\notin F^i_1$. In all other cases $e\in F^i_1$. Similarly, if $f\in M_i$ and $f$ is outgoing from $Q$, then $e\notin F^i_2$. Also, if $g\in M_i$ and $g$ is outgoing from $Q$, then $g\notin F^i_2$. In all other case $e\in F^i_2$. We conclude
\begin{align*}
\sum_{i=1}^{k}\frac{\lambda_i}{2}\chi^{ F^i_1}_e+ \sum_{i=1}^{k}\frac{\lambda_i}{2}\chi^{ F^i_2}_e &=\;\;\frac{1}{2}- \sum_{i\in k:\;f\in M_i, f\rightarrow Q}\frac{\lambda_i}{2}- \sum_{i\in k:\;g\in M_i, g\rightarrow Q}\frac{\lambda_i}{2}\\
&\;\;+\frac{1}{2}- \sum_{i\in k:\;f\in M_i, f\leftarrow Q}\frac{\lambda_i}{2}- \sum_{i\in k:\;g\in M_i, g\leftarrow Q}\frac{\lambda_i}{2}\\
&=\;\;1- \sum_{i\in k:\;f\in M_i}\frac{\lambda_i}{2}- \sum_{i\in k:\;g\in M_i}\frac{\lambda_i}{2}\\
&=\;\; 1- \alpha.
\end{align*}
\end{cproof}
This concludes the proof.
\end{proof}

\subsection{A top-down coloring approach: Proof of Lemma \ref{9/10}}\label{sec:9/10}
In this section we prove Lemma \ref{9/10}.

\nineoverten*

In order to prove this lemma, we
need a way to reduce vertex connectivity to edge-connectivity to be able to employ the top-down coloring approach. 

\subsubsection{Reducing 2-vertex connectivity to 2-edge
	connectivity}\label{vertex-cut}

We now present an approach to reduce vertex connectivity to
edge-connectivity.  Let $G=(V,E)$ be a 4-regular 4-edge-connected
graph.  Note that $G$ must be 2-vertex-connected.  Let $T$ be a
spanning tree of $G$ such that $T$ does not have any vertices of
degree four and let $L=E\setminus T$ be the set of links. 

For a link $\ell$ in $L$, let $P_\ell$ be the set of edges in $T$ on
the unique path in $T$ between the endpoints of $\ell$. For $e\in T$,
let $ \cov(e)$ be the set of links $\ell$ such that $e\in
P_\ell$. Since $G$ is 4-edge-connected, $|\cov(e)|\geq 3$ for all
$e\in T$.

\begin{definition}\label{subdivide}
The {\em subdivided graph} $G'=(V',E')$ of
$G$ is the graph in which each edge $e=uw$ of $T$ is subdivided into
$uv_e$ and $v_ew$.  Then $T'$ is a spanning tree of $G'$ in which for
each edge $uw \in T$, we include both $uv_e$ and $v_ew$ in $T'$.  We
define $L' = E'\setminus{T'}$ as follows.  For each link $\ell \in L$,
we make a link $\ell' \in L'$ as follows.  Let $u$ be an endpoint of
$\ell$.
\begin{enumerate}

\item If $u$ is a leaf of $T$, then $u$ is an endpoint of $\ell'$.

\item If $u$
        is an internal vertex, let $e$ be the edge in $P_\ell$ such
        that $u$ is also an endpoint of $e$.  (Note that there is only
        one such $e$, since $P_\ell$ is a unique path and $e$ is the
        first, or last, edge in $P_\ell$.)  Then $v_e$ is the endpoint
        of $\ell'$.

\end{enumerate}
\end{definition}

The procedure outlined in Definition \ref{subdivide} defines a
bijection between links in $L$ and $L'$.  Thus, for every set of links
$F' \subset L'$, we let $F \subset L$ denote the corresponding set of
links.  We use this bijection to go from 2-edge-connectivity to
2-vertex-connectivity.

\begin{lemma}\label{reductionlemma}
Let $G=(V,E)$ be a 4-regular 4-edge-connected graph, let $T$ be a
spanning tree of $G$ with maximum degree three, and let $L =
E\setminus{T}$.  Let $G' = (V',E')$ be a subdivided graph with
spanning tree $T'$ and links $L' = E'\setminus{T'}$.  We have
  
\begin{itemize}
	\item  For any $F' \subset L'$ such that $T' + F'$ is
          2-edge-connected, $T + F$ is 2-vertex-connected.

	\item For every edge $e'\in T'$, there are at least two links
          $\ell'_1,\ell'_2\in L'$ such that $\ell'_1,\ell'_2 \in \cov(e')$.
\end{itemize}
\end{lemma}

\begin{proof}
	Let us show that this reduction satisfies the first
        property. Suppose for contradiction that there is $F'\subseteq
        L'$ such that $T'+F'$ is 2-edge-connected, but the
        corresponding set of links $F$, is such that $T+F$ has a
        cut-vertex, namely $u$. Clearly $u$ cannot be a leaf of $T$,
        since $T-u$ is a connected graph. Similarly, $r\neq u$. Hence,
        we can assume that $u$ is an internal vertex of $T$.

Since $u$ is a cut-vertex of $T+F$, we can partition $V\setminus
\{u\}$ into $S_1$ and $S_2$ such that there is no edge in $T+F-
\delta(u)$ that has one endpoint in $S_1$ and one endpoint in $S_2$.
Let $\delta_T(u)$ be the set of edges in $T$ incident on $u$.  Since
$u$ is an internal vertex of $T$, we have $2\leq|\delta_T(u)| \leq 3$.
Suppose $u$ is adjacent to $v$ in $T$. Label the $vu$ edge in $T$ with
$e$.  Assume first that $|\delta_T(u)|=2$: let $f$ be the other edge
incident to $u$ in $T$. There is no link $\ell'\in F'$ such that
$\ell'$ covers the edge $uv_f$, because such a link $\ell'$
corresponds to a link in $\ell \in L$ that has one endpoint in $S_1$
and other in $S_2$. Now, assume $|\delta_T(u)|=3$: let $f_1$ and $f_2$
be the edges incident to $u$ (besides $e$) in $T$. Let $w_1$ and $w_2$
be the endpoints of $f_1$ and $f_2$ other than $u$. Again, let $S_1$
and $S_2$ be a partition of $V\setminus \{u\}$ such that no edge in
$T+F-\delta(u)$ that has one end in $S_1$ and other in $S_2$. Without
loss of generality, assume $v\in S_1$ and $w_1,w_2\in S_2$. Consider
edge $v_eu$ in $T'$: if there is a link $\ell'\in L'$ covering $v_eu$,
then the link $\ell$ corresponding to $\ell'$ has one end in $S_1$ and
the other in $S_2$. Hence, we get a contradiction.

Now we show the second property holds: for each edge $e'\in T'$, there
are at least two links $\ell,\ell' \in L'$ that are in
$\cov(e')$. Suppose there is an edge $e'$ such that $e'$ does not have
this property. Edge $e'$ corresponds to one part of a subdivided edge
$e$ in the tree $T$. Let $v$ and $v_e$ be the endpoints of $e'$.

First, notice that if $v$ is a leaf,
then there are three links in $\ell$ that cover edge $e$ in
$T$, and all these links will cover $e'$ in the new instance as we do not
change the leaf endpoints. Thus we may assume that $v$ is not a leaf.

	\begin{figure}[h]
		\centering
		\begin{subfigure}[t]{.3\textwidth}\centering
			\begin{tikzpicture}[scale=0.9]
			\draw [-] [black, line width=0.4mm] plot [smooth, tension=0] coordinates {(1,2.5) (1,1)};
			
			\draw [-] [black, line width=0.4mm] plot [smooth, tension=0] coordinates {(1,1) (1,-0.5)};

			\draw [dashed] [red, line width=0.4mm,xshift=0 cm] plot [smooth, tension=1] coordinates {(1,1) (0.4,1.7) (0.1,2.4)};

			\draw [dashed] [red, line width=0.4mm,xshift=0 cm] plot [smooth, tension=1] coordinates {(1,1) (1.6,1.7) (1.9,2.4)};

			\draw [dashed] [red, line width=0.4mm,xshift=0 cm] plot [smooth, tension=1] coordinates {(1,1) (0.4,0.3) (0.1,-0.4)};

			\draw [dashed] [red, line width=0.4mm,xshift=0 cm] plot [smooth, tension=1] coordinates {(1,1) (1.6,0.3) (1.9,-0.4)};

			\draw[black,fill=white] (1,-0.5) ellipse (0.1 cm  and 0.1 cm);
			\draw[black,fill=white] (1,1) ellipse (0.1 cm  and 0.1 cm);
			\draw[black,fill=white] (1,2.5) ellipse (0.1 cm  and 0.1 cm);
			
			\node (Q) at (1.2,1.8) {{$e$}};
			\node (Q) at (1.2,0.2) {{$f$}};
			\node (Q) at (1.3,1) {{$v$}};
			\end{tikzpicture}
			\caption{}\label{fig:5.1a}
		\end{subfigure}
		\begin{subfigure}[t]{.3\textwidth}\centering
			\begin{tikzpicture}[scale=0.9]
			
			\draw [-] [black, line width=0.4mm] plot [smooth, tension=0] coordinates {(1,2.5) (1,1)};
			
			\draw [-] [black, line width=0.4mm] plot [smooth, tension=0] coordinates {(1,1) (0,0)};
			
			\draw [-] [black, line width=0.4mm] plot [smooth, tension=0] coordinates {(1,1) (2,0)};
			
			\draw [dashed] [red, line width=0.4mm,xshift=0 cm] plot [smooth, tension=1] coordinates {(1,1) (0.4,1.7) (0.1,2.4)};

			\draw [dashed] [red, line width=0.4mm,xshift=0 cm] plot [smooth, tension=1] coordinates {(1,1) (1.6,1.7) (1.9,2.4)};

			\draw[black,fill=white] (0,0) ellipse (0.1 cm  and 0.1 cm);
			\draw[black,fill=white] (2,0) ellipse (0.1 cm  and 0.1 cm);
			\draw[black,fill=white] (1,1) ellipse (0.1 cm  and 0.1 cm);
			\draw[black,fill=white] (1,2.5) ellipse (0.1 cm  and 0.1 cm);
		
			\node (Q) at (1.2,1.8) {{$e$}};
			\node (Q) at (0.35,0.6) {{$g$}};
			\node (Q) at (1.72,0.6) {{$f$}};
			\end{tikzpicture}
			\caption{}		\label{fig:5.1b}
		\end{subfigure}
	\caption{Different cases for the second property in Lemma \ref{reductionlemma}.}
	\end{figure}

If $v$ has degree two in $T$, then let edge $f$ be the other edge
incident to $v$, as shown in Figure \ref{fig:5.1a}.  Let $\ell$ be a
link in $L$ such that $e$ and $f$ are both covered by $\ell$. If
$\ell'\in L'$ is the link corresponding to $\ell$, then $\ell'$ covers
$e'$. Hence we can suppose there is at most one link $\ell$ in $L$
that covers both $e$ and $f$. Therefore, there are distinct links
$\ell_1,\ldots,\ell_4$ such that $\ell_1,\ell_2$ cover $f$ and
$\ell_3,\ell_4$ cover $e$. But then vertex $v$ has degree six in $G$
as every link that covers $e$ and does not cover $f$ or vice versa
must have $v$ as an endpoint. Thus, we may assume that $v$ has degree
three in $T$, which means $v$ is incident to edges $e,f$ and $g$ in
$T$, as shown in Figure \ref{fig:5.1b}.  Let $\ell_1,\ell_2,\ell_3$ be
the links that cover $e$. Suppose without loss of generality that
$\ell_1$ and $\ell_2$ cover either $f$ or $g$. Then, the corresponding
links $\ell_1'$ and $\ell_2'$ in $L'$ will cover $e'$. However, if
$\ell_1$ does not cover $f$ or $g$ if must be the case that $\ell_1$
has an endpoint in $v$. The same holds for $\ell_2$. This implies that
$v$ has degree five, which is a contradiction to the 4-regularity of
$G$.
\end{proof}

\subsubsection{The top-down coloring algorithm}\label{app:coloring}


We want to find a set of links $F' \subset L'$ such that i) $T' + F'$
is 2-edge-connected, and ii) each vertex in $T + F$ has degree at
least three.  Now we expand our terminology for a top-down coloring
algorithm to address these additional requirements.  For each
$\ell'\in L'$, where $\ell$ is the link in $L$ corresponding to
$\ell'$, we define $\point(\ell')$ to be the two endpoints of $\ell$
in $G$.

For a $(\p,q)$-coloring of $L'$, namely $\gamma$, we say that $v$ in
$G$ has received a color $c$ in $\gamma$ if there is $\ell'$ such that
$v\in \point(\ell')$ and $c\in \gamma(\ell')$. We say a vertex $v$
received a color $c$ twice in $\gamma$, if there are two links $\ell'$
and $\ell''$ such that $v\in \point(\ell')$ and $v\in \point(\ell'')$
and both $c\in \gamma(\ell')$ and $c\in\gamma(\ell')$. Similarly, we
say \textit{$v$ is missing color $c$ in $\gamma$}, if there is no link
$\ell'$ such that $v\in \point(\ell')$ with $c\in
\gamma(\ell)$. Moreover, we say $v$ is \textit{missing a color $c$ for
  the second time in $\gamma$}, if there is exactly one link $\ell'$
with $v\in \point(\ell')$ with $c\in\gamma(\ell')$.

\begin{lemma}\label{vertexcolor}
Let $G=(V,E)$ be a 4-regular 4-edge-connected graph and let $T$ be a
spanning tree of $G$ with maximum degree three.  Let $G'=(V',E')$ and
$T'$ be the subdivided graph and spanning tree.  Then there is
$T'$-admissible $(4\chi^{L'},5)$-coloring of $L'=E'\setminus T'$,
namely $\gamma$ such that for a vertex $v$ of $G$, i) if $v$ has
degree two in $T$, then $v$ receives all the five colors in $\gamma$,
and ii) if $v$ is a degree one vertex in $T$, then $v$ receives all
the five colors twice in $\gamma$.
\end{lemma}

\begin{proof}
We construct $\gamma$ using the top-down coloring algorithm.

Let $\gamma$ be the $(4\chi^{L'},5)$-coloring of $L'$ that we maintain. Initially, we have $\gamma(\ell')=\emptyset$ for $\ell' \in L'$. Suppose we want to color link $\ell'$ at some iteration of the algorithm.  Let $u',v'$ be the endpoints of $\ell'$ in $G'$. Let $s'$ be the LCA of $\ell'$ in $T'$. Let $\Lup$ be the $s'u'$-path in $T'$ and $\Rvp$ be the $s'v'$-path in $T'$. Let $\point(\ell')= \{u,v\}$.

\paragraph{Coloring Rules:}\begin{enumerate}
		\item If there is a color $c$ that $u$ has not received we set one color on $\ell'$ to be $c$. If $u$ is not missing a color, but missing a color $c$ for the second time, give color $c$ to $\ell'$.
		\item If there is a color $c$ that $v$ has not received we set one color on $\ell'$ to be $c$. If $v$ is not missing a color, but missing a color $c$ for the second time, give color $c$ to $\ell'$.

		\item  Give color $\bar{\gamma}_1(\Lup)$ to
                $\ell'$.  If there is no such color and vertex $u$ is
                missing a color $c$ for the second time, give color
                $c$ to $\ell'$.

\item Give color $\bar{\gamma}_1(\Rvp)$ to $\ell'$.  If there is no
  such color and vertex $v$ is missing a color $c$ for the second
  time, give color $c$ to $\ell'$.

		\item If after applying all the above four rules, 
                  $\ell'$ still has fewer than four distinct colors, give $\ell'$ any color that it does not already have until $\ell'$ has four different colors.
	\end{enumerate}
	\vspace{2mm}
	First we show that the top-down coloring algorithm above is
        $T'$-admissible. Consider an edge $e'$ in $T'$. We know by Lemma
        \ref{reductionlemma} that there are links $\ell'$ and $\ell''$
        in $L'$ such that $\ell',\ell''\in \cov(e')$.  Without loss of
        generality, suppose that $\ell'$ has a higher LCA.  After we
        color $\ell'$, $e'$ has received at least four colors
        (Observation \ref{firstlink}). When we color $\ell''$ (using
        Rule 3 or 4) we give at least one new color to $e'$ so it receives all the five colors (Observation \ref{ansectory}). Therefore, the coloring algorithm is $T'$-admissible. 
	
	Now, we show the extra properties hold as well. Consider a
        vertex $v$ of degree two in $T$. Notice that since $G$ is
        4-regular, there are at least two links $\ell'$ and $\ell''$
        such that $v\in \point(\ell')$ and $v\in \point(\ell'')$. At
        the iteration the algorithm colors $\ell'$, vertex $v$
        receives four new colors, and later when the algorithm color
        $\ell''$, vertex $v$ receives its fifth missing color (by Rule
        1 or 2). 
	
	Finally, assume $v$ is a vertex of degree one in $T$. This
        implies that $v'$ is also a degree one vertex in $T'$ (since
        in the reduction we do not change the endpoints for degree one
        vertices). Let $e'_{v'}$ be the leaf edge in $T'$ incident on
        $v'$. By 4-regularity there are three links
        $\ell'_1,\ell'_2,\ell'_3$ labeled in LCA order  such that
        $v\in \point(\ell'_i)$ for $i\in\{1,2,3\}$. In the iteration
        that $\ell'_1$ is colored, $v$ receives four new
        colors. Later, when $\ell'_2$ is colored, $v$ receives its
        last missing color (by Rule 1 or 2). In other words, after
        coloring $\ell'_2$, vertex $v$ has received all five colors
        and has received three colors twice. This means that after
        coloring $\ell'_2$, vertex $v$ is missing exactly two colors
        for the second time. Furthermore, $\ell'_1,\ell'_2 \in
        \cov(e'_{v'})$. This implies by the argument above, when the
        algorithm colors $\ell'_2$, edge $e'_{v'}$ has received all
        the five colors. Consider the time when the algorithm wants to
        color $\ell'_3$. Notice that all the ancestors of $e'_{v'}$
        have received all the five colors, and $e'_{v'}$ is the lowest
        edge in $\mathcal{R}_{\ell'_3}$. Therefore, there is no
        missing color in $\mathcal{R}_{\ell'_3}$. Also, $v$ has
        received all five colors. Therefore, when coloring $\ell'_3$,
        vertex $v$ will receives two new colors for the second time
        (where the first color is assigned by Rule 1 or 2 and the
        second color by Rule 3 or 4).  
\end{proof}

Lemma \ref{9/10} follows immediately from  Lemmas \ref{reductionlemma} and \ref{vertexcolor}.
\subsection{Hard to round half-integer square points}\label{hard-to-round}

As discussed in the introduction, $\alpha^{\HK}_{\EC}\geq \frac{6}{5}$.
An example achieving this lower bound is given in 
\cite{alexander2006integrality}.  However, a more curious instance is
the $k$-donut. A $k$-donut point for $k\in \mathbb{Z}$, $k\geq 2$, is
a graph $G_k=(V_k,E_k)$ that has $k$ half-squares arranged around a
cycle, and the squares are joined by paths consisting of $k$ 1-edges.
(See Figure \ref{k-donut} for an illustration of the 4-donut.)

\begin{figure}[t]
	\centering
		\includegraphics[width=0.25\linewidth]{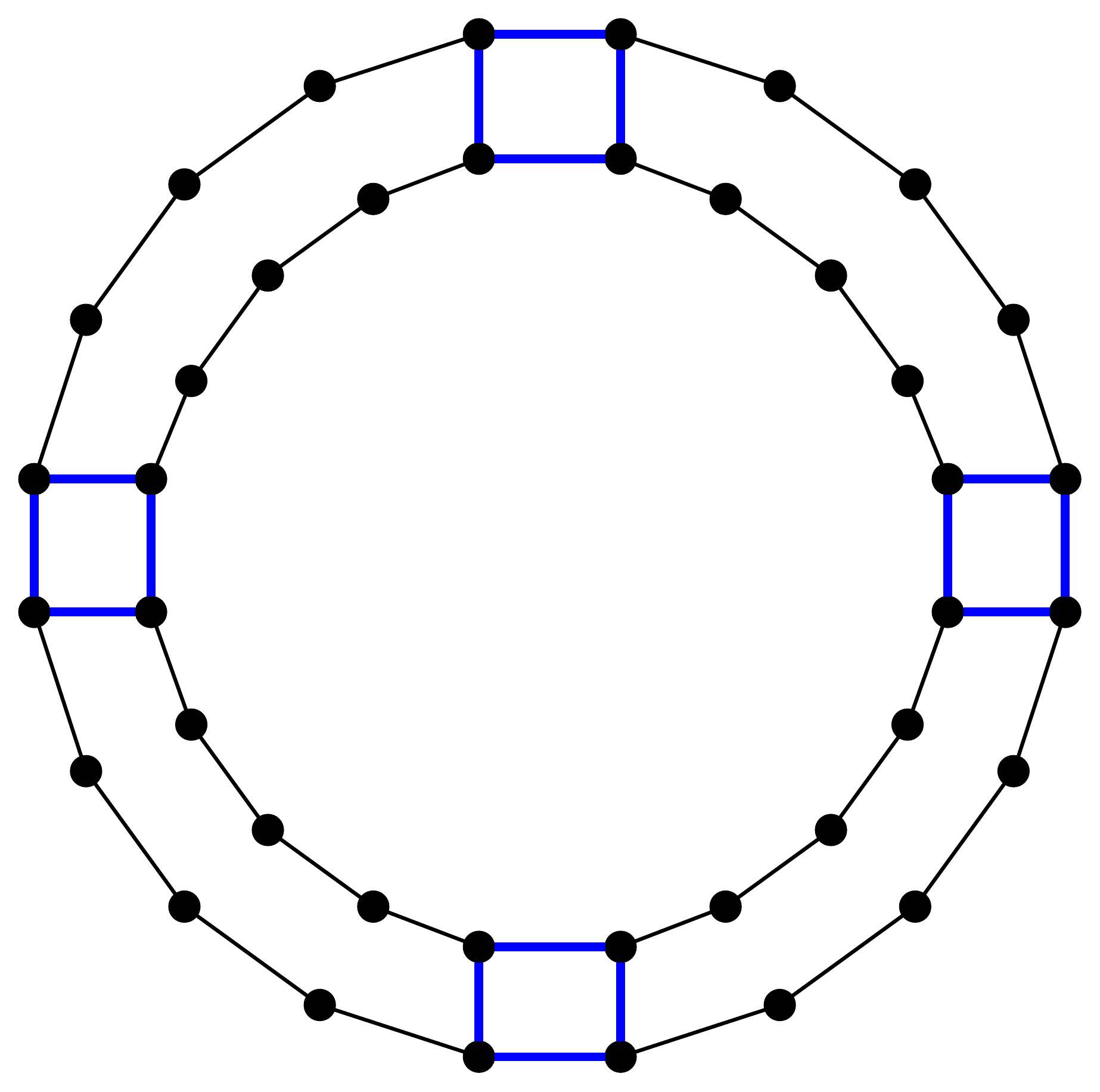}
		\caption{The $k$-donut for $k=4$: bold (blue) edges
                  are the half edges and remaining edges are 1-edges.}
	\label{k-donut}
\end{figure}
Define the edge cost $c_e$ of each half-edge in the outer cycle and
the inner cycle to be 2.  All other half-edges have cost 1. All the
1-edges have cost $\frac{1}{k}$. Therefore, $\sum_{e\in E_x} c_ex_e =
5k$, while the optimal solution is $6k-2$.  We note that this instance
was discovered by the authors of \cite{Carr98}, but due to the page
limit of their conference paper they did not present it and just
mentioned that they know a lower bound.  Recently, Boyd and Seb\H{o}
used $k$-donut points with different costs to show a new instance that
achieves a lower bound of $\frac{4}{3}$ for the integrality gap of 2ECM
and TSP, and we attribute the term ``$k$-donut'' to
them~\cite{boydsebo-mp}.  Notice that if $x$ is the $k$-donut point,
then $P(G,\frac{1}{4}-\frac{1}{2k})$ holds.  This implies that $z=
\frac{1}{5}\chi^{H}+ \frac{4}{5}r^{x,\frac{1}{4}-\frac{1}{2k}}$ is a
convex combination of 2-edge-connected multigraphs of $G_x$.  We have
$z_e= \frac{6}{5}-\frac{2}{5k}$ for $e\in A$, $z_e = \frac{3}{5}$ for
$e\in B$, and $z_e = \frac{3}{5}+\frac{2}{5k}$ for $e \in C$.  As $k
\rightarrow \infty$, this approaches $\frac{6}{5}x$ and thus shows
that our approach can verify the six-fifths conjecture for $k$-donut
points.  We conclude with the following corollary of Theorem
\ref{square-point}.
\begin{corollary}
The integrality gap $\alpha^{\LP(G)}_{\EC}$ is between $\frac{6}{5}$
and $\frac{9}{7}$ for half-integer square points.
\end{corollary}

\section*{Acknowledgments}
We would like to thank R. Ravi for his insightful comments throughout
this project.  We would also like to thank Robert Carr for sharing the
$k$-donut example.  AH was supported in part by the U.S. Office of
Naval Research under award number N00014-18-1-2099, and the U.S.
National Science Foundation under award number CCF-1527032.  AN was
supported in part by IDEX-IRS SACRE.

\bibliographystyle{unsrt}
\bibliography{cubic3EC}

\newcommand{\etalchar}[1]{$^{#1}$}
\begin{thebibliography}{AGM{\etalchar{+}}17}

\bibitem[ABE06]{alexander2006integrality}
Anthony Alexander, Sylvia Boyd, and Paul Elliott{-}Magwood.
\newblock On the integrality gap of the 2-edge connected subgraph problem.
\newblock Technical report, TR-2006-04, SITE, University of Ottawa, 2006.

\bibitem[AGM{\etalchar{+}}17]{asadpour2017log}
Arash Asadpour, Michel~X Goemans, Aleksander Madry, Shayan {Oveis Gharan}, and
  Amin Saberi.
\newblock An ${O}(\log{n}/\log \log{n})$-approximation algorithm for the
  asymmetric traveling salesman problem.
\newblock {\em Operations Research}, 65(4):1043--1061, 2017.

\bibitem[BC11]{Boyd-Carr-2011}
Sylvia Boyd and Robert Carr.
\newblock Finding low cost {TSP} and 2-matching solutions using certain
  half-integer subtour vertices.
\newblock {\em Discrete Optimization}, 8(4):525--539, 2011.

\bibitem[BCC{\etalchar{+}}20]{boyd_et_al:LIPIcs:2020:12664}
Sylvia Boyd, Joseph Cheriyan, Robert Cummings, Logan Grout, Sharat Ibrahimpur,
  Zolt{\'a}n Szigeti, and Lu~Wang.
\newblock A 4/3-approximation algorithm for the minimum 2-edge connected
  multisubgraph problem in the half-integral case.
\newblock In {\em International Conference on Approximation Algorithms for
  Combinatorial Optimization Problems (APPROX)}, volume 176, pages 61:1--61:12,
  2020.

\bibitem[BFS16]{boyd2014frac}
Sylvia Boyd, Yao Fu, and Yu~Sun.
\newblock A 5/4-approximation for subcubic 2{EC} using circulations and obliged
  edges.
\newblock {\em Discrete Applied Mathematics}, 209:48--58, 2016.

\bibitem[BIT13]{bit13}
Sylvia Boyd, Satoru Iwata, and Kenjiro Takazawa.
\newblock Finding 2-factors closer to {TSP} tours in cubic graphs.
\newblock {\em SIAM Journal on Discrete Mathematics}, 27(2):918--939, 2013.

\bibitem[BL97]{broersma1997spanning}
Hajo Broersma and Xueliang Li.
\newblock Spanning trees with many or few colors in edge-colored graphs.
\newblock {\em Discussiones Mathematicae Graph Theory}, 17(2):259--269, 1997.

\bibitem[BL17]{philip6/7}
Sylvia Boyd and Philippe Legault.
\newblock Toward a 6/5 bound for the minimum cost 2-edge connected spanning
  subgraph.
\newblock {\em SIAM Journal on Discrete Mathematics}, 31(1):632--644, 2017.

\bibitem[BS21]{boydsebo-mp}
Sylvia Boyd and Andr{\'a}s Seb{\H{o}}.
\newblock The salesman’s improved tours for fundamental classes.
\newblock {\em Mathematical Programming}, 186:289--307, 2021.

\bibitem[Car11]{caratheodory1911variabilitatsbereich}
Constantin Carath{\'e}odory.
\newblock {\"U}ber den {V}ariabilit{\"a}tsbereich der {F}ourierschen
  {K}onstanten von positiven harmonischen {F}unktionen.
\newblock {\em Rendiconti Del Circolo Matematico di Palermo}, 32(2):193--217,
  1911.

\bibitem[Chr76]{christofides}
Nicos Christofides.
\newblock {Worst-case analysis of a new heuristic for the travelling salesman
  problem}.
\newblock Technical Report 388, Graduate School of Industrial Administration,
  Carnegie Mellon University, 1976.

\bibitem[CJR99]{cheriyan19992}
Joseph Cheriyan, Tibor Jord{\'a}n, and R.~Ravi.
\newblock On 2-coverings and 2-packings of laminar families.
\newblock In {\em European Symposium on Algorithms}, pages 510--520. Springer,
  1999.

\bibitem[CR98]{Carr98}
Robert Carr and R.~Ravi.
\newblock A new bound for the 2-edge connected subgraph problem.
\newblock In {\em Proceedings of 6th International Conference on Integer
  Programming and Combinatorial Optimization}, pages 112--125. Springer, 1998.

\bibitem[CV04]{carrvempala}
Robert Carr and Santosh Vempala.
\newblock On the {H}eld-{K}arp relaxation for the asymmetric and symmetric
  traveling salesman problems.
\newblock {\em Mathematical Programming}, 100(3):569--587, 2004.

\bibitem[DJS03]{devos2003cut}
Matt DeVos, Thor Johnson, and Paul Seymour.
\newblock Cut coloring and circuit covering.
\newblock
  \url{https://web.math.princeton.edu/~pds/papers/cutcolouring/paper.pdf},
  2003.
\newblock Unpublished manuscript.

\bibitem[Gao13]{gao2013lp}
Zhihan Gao.
\newblock An {LP}-based-approximation algorithm for the $s$-$t$ path graph
  traveling salesman problem.
\newblock {\em Operations Research Letters}, 41(6):615--617, 2013.

\bibitem[Goe95]{Goemans95}
Michel~X. Goemans.
\newblock Worst-case comparison of valid inequalities for the {TSP}.
\newblock {\em Mathematical Programming}, 69(1):335--349, 1995.

\bibitem[Had20]{arash-phd2020}
Arash Haddadan.
\newblock {\em {New Bounds on Integrality Gaps by Constructing Convex
  Combinations}}.
\newblock PhD thesis, Carnegie Mellon University, 2020.

\bibitem[HN20]{HN-extended}
Arash Haddadan and Alantha Newman.
\newblock Towards improving {C}hristofides algorithm on fundamental classes by
  gluing convex combinations of tours.
\newblock {\em CoRR}, abs/1907.02120, 2020.

\bibitem[HNR21]{uniform}
Arash Haddadan, Alantha Newman, and R.~Ravi.
\newblock Shorter tours and longer detours: {U}niform covers and a bit beyond.
\newblock {\em Mathematical Programming}, 185:245--273, 2021.

\bibitem[HS21]{horscheulerian}
Florian H{\"o}rsch and Zolt{\'a}n Szigeti.
\newblock Connectivity of orientations of 3-edge-connected graphs.
\newblock {\em European Journal of Combinatorics}, 2021.
\newblock To appear.

\bibitem[IR17]{IR}
Jennifer Iglesias and R.~Ravi.
\newblock Coloring down: $3/2$-approximation for special cases of the weighted
  tree augmentation problem.
\newblock {\em CoRR}, abs/1707.05240, 2017.

\bibitem[KN{\etalchar{+}}05]{kaiser2005unions}
Tom{\'a}{\v{s}} Kaiser, Serguei Norine, et~al.
\newblock Unions of perfect matchings in cubic graphs.
\newblock {\em Electronic Notes in Discrete Mathematics}, 22:341--345, 2005.

\bibitem[Leg17]{philip}
Philippe Legault.
\newblock {Towards New Bounds for the 2-Edge Connected Spanning Subgraph
  Problem}.
\newblock Master's thesis, University of Ottawa, 2017.

\bibitem[NP81]{naddef1981matchings}
Denis Naddef and William~R. Pulleyblank.
\newblock Matchings in regular graphs.
\newblock {\em Discrete Mathematics}, 34(3):283--291, 1981.

\bibitem[NW61]{nw}
C.~St.J.~A. Nash-Williams.
\newblock Edge-disjoint spanning trees of finite graphs.
\newblock {\em Journal of the London Mathematical Society}, s1-36(1):445--450,
  1961.

\bibitem[SV14]{sebo12}
Andr{\'a}s Seb{\H{o}} and Jens Vygen.
\newblock Shorter tours by nicer ears: 7/5-approximation for the graph-{TSP},
  3/2 for the path version, and 4/3 for two-edge-connected subgraphs.
\newblock {\em Combinatorica}, 34(5):597--629, 2014.

\bibitem[SWvZ14]{schalekamp20142}
Frans Schalekamp, David~P Williamson, and Anke van Zuylen.
\newblock 2-{M}atchings, the traveling salesman problem, and the subtour {LP}:
  {A} proof of the {B}oyd-{C}arr conjecture.
\newblock {\em Mathematics of Operations Research}, 39(2):403--417, 2014.

\bibitem[Tak16]{takazawa}
Kenjiro Takazawa.
\newblock A 7/6-approximation algorithm for the minimum 2-edge connected
  subgraph problem in bipartite cubic graphs.
\newblock {\em Information Processing Letters}, 116(9):550 -- 553, 2016.

\bibitem[Wol80]{wolsey1980heuristic}
Laurence~A. Wolsey.
\newblock Heuristic analysis, linear programming and branch and bound.
\newblock In {\em Combinatorial Optimization II}, pages 121--134. Springer,
  1980.

\end{thebibliography}
\appendix
\section{Proof of Theorem \ref{reductionToFourEC}}\label{app:gluing}

The main tool used to prove Theorem \ref{reductionToFourEC} is {\em
  gluing solutions over 3-edge cuts}~\cite{Carr98}, which was
introduced by Carr and Ravi~\cite{Carr98} and has been used frequently
for constructing convex combinations of 2-edge-connected
subgraphs~\cite{Carr98,philip6/7,philip} and more recently for gluing
tours (see Proposition 3.2 in \cite{HN-extended}).

We need the following definitions.

\begin{definition}
A {\em proper 3-edge cut} of $G$ is a set $S\subset V$ such that
$|\delta(S)|=3$, $|S|\geq 2$ and $|V\setminus{S}| \geq 2$.
\end{definition}

\begin{definition}\label{G-S-graph}
	For a 3-edge-connected graph $G=(V,E)$ and a proper 3-edge cut
        $S \subset V$, we define the graph $G_S$ to be the graph
        obtained after we contract the set
$V\setminus{S}$ to a single vertex.
\end{definition}

Lemmas \ref{gluing-in-polytime} and \ref{gluing-essentially-4EC} appear in
different forms in \cite{Carr98,philip6/7,philip}, always with the
purpose of reducing to the problem on essentially 4-edge-connected
graphs.

\begin{restatable}{lemma}{gluingX}\label{gluing-in-polytime}
Let $G=(V,E)$ be a 3-edge-connected graph and $x\in [0,1]^E$. Let $S$
be a 3-edge cut of $G$. Define $x^S$ and $x^{\bar{S}}$ to be vector
$x$ restricted to the edges in $G_S$ and $G_{\bar{S}}$, respectively.
If $x^S$ and $x^{\bar{S}}$ can be written as a convex combination of
2-edge-connected subgraphs of $G_S$ and $G_{\bar{S}}$, respectively,
then $x$ can be written as convex combination of 2-edge-connected
subgraphs of $G$.
\end{restatable}	
      
\begin{proof}
By the assumption, vector $x^S$ can be written as a convex combination
of 2-edge-connected subgraphs of $G_S$: $x^S_e =
\sum_{i=1}^{k}\lambda_i \chi^{F_S^i}_e$ for $e \in E(G_S)$.  The same
holds for $G_{\bar{S}}$: $x^{\bar{S}}_e = \sum_{i=1}^{\bar{k}}\theta_i
\chi^{F_{\bar{S}}^i}_e$ for $e \in E(G_{\bar{S}})$.

Note that $\delta(X_S)= \delta(X_{\bar{S}}) = \{a,b,c\}$, and hence
$x^S_e = x^{\bar{S}}_e = x_e$ for $e\in \{a,b,c\}$.  Let
$\lambda^{a,b}$ be the sum of all $\lambda_i$'s where $F_S^i$ contains
exactly the two edges $a$ and $b$ from $\delta(X_S)$. Define
$\lambda^{a,c}, \lambda^{b,c}$, and $\lambda^{a,b,c}$ analogously.
Notice that these are the only possible outcomes since a
2-edge-connected subgraphs contains at least two edges from the cut
around any vertex.  Hence, $\lambda^{a,b}+\lambda^{a,c} +
\lambda^{b,c} + \lambda^{a,b,c}=1$. Also
\begin{align*} \lambda^{a,b}+
\lambda^{a,c}+\lambda^{a,b,c} = x_a,\\ \lambda^{a,b}+
\lambda^{b,c}+\lambda^{a,b,c} = x_b,\\ \lambda^{a,c}+
\lambda^{b,c}+\lambda^{a,b,c} = x_c. \end{align*} This system of
equations has a unique solution: $\lambda^{a,b,c} = x_a+ x_b+x_c- 2$,
$\lambda^{b,c}=1-x_a$, $\lambda^{a,b}=1-x_c$, and $\lambda^{a,c}= 1-
x_b$.  Similarly, we can define and show that $\theta^{a,b,c}=
x_a+x_b+x_c - 2$, $\theta^{b,c}=1-x_a$, $\theta^{a,b}=1-x_c$, and
$\theta^{a,c}= 1- x_b$.
		
So we have $\lambda^{h} = \theta^{h}$ for $h\in
\{\{a,b\},\{a,c\},\{b,c\},\{a,b,c\}\}$. This allows us to glue the two
convex combinations in the following way: suppose $F_S^i$ and
$F_{\bar{S}}^j$ use the same edges from $\{a,b,c\}$.  Now we glue
$\sum_{i=1}^{k}\lambda_i \chi^{F_S^i}$ and
$\sum_{i=1}^{\bar{k}}\theta_i \chi^{F_{\bar{S}}^i}$ as follows.  Let
$\sigma_{ij} = \min\{\lambda_i,\theta_j\}$, and $F^{ij}=F_S^i +
F_{\bar{S}}^j$. Update $\lambda_i$ and $\theta_j$ by subtracting
$\sigma_{ij}$ from both, and continue. The arguments in the lemma
ensure that we can find the $i$ and $j$ pair until all the remaining
$\lambda_i$ and $\theta_j$ multipliers are zero. The convex
combination with multipliers $\sigma_{ij}$ and 2-edge-connected
subgraphs $F^{ij}$ is equal to $x_e$ on every edge in $E(G)$.  Note
that the number of subgraphs in the set
$\{F^{ij}\}$ (i.e., the support of the new convex combination) is
at most $k+\bar{k}$.  Assuming that the size of the support of the convex
combinations for both $G_S$ and $G_{\bar{S}}$ are polynomial in the
size of the respective graphs, then
the total number of subgraphs in the convex combination produced for $G$ is
polynomial.
\end{proof}

The next two observations are needed for the proof of the next lemma.

\begin{observation}\label{obs:cornercut}
Let $U\subseteq V$ be a minimal proper 3-edge-cut in $G$ (i.e., for
$S\subset U$, the cut defined by $S$ is not a proper 3-edge cut in
$G_{U}$). Then $G_{U}$ does not contain any proper 3-edge cuts.
\end{observation}
\begin{proof}
Suppose for contradiction that there is $S \subset V(G_{U})$ that is a
proper 3-edge cut of $G_{U}$.  This is a contradiction to minimality
of $U$ since $S$ constitutes a proper 3-edge cut in $G$ as well.
\end{proof}

\begin{observation}\label{obs:inductivecriticalcut}
Suppose that $G$ has $t$ proper 3-edge cuts. Let $U$ be a proper 3-edge cut of
$G$. Then the number of proper 3-edge cuts in $G_{U}$ is at most $t-1$.
\end{observation}
\begin{proof}

There is correspondence between the proper 3-edge cuts of $G_{U}$ and
$G$: each proper 3-edge cut in $G_U$ corresponds to a proper 3-edge
cut in $G$.  However, clearly, $U$ is not a proper 3-edge cut of
$G_{U}$.  This implies that $G_{U}$ can have at most $t-1$ proper
3-edge cuts. 
\end{proof}

\begin{definition}\label{G-family}
	Let $G=(V,E)$ be a 3-edge-connected graph.
        Define $\mathcal{G}$ to be the collection of graphs
        obtained from $G$ by iteratively contracting an arbitrary
        proper 3-edge cut and contracting it into a single vertex
        until the graph becomes essentially 4-edge-connected.  
        \end{definition}

\begin{restatable}{lemma}{gluingEssentiallyFourEC}\label{gluing-essentially-4EC}
	Let $G=(V,E)$ be a 3-edge-connected graph and $x\in [0,1]^E$.
        The following two statements are equivalent.
\begin{enumerate}

\item For any $G'\in \mathcal{G}$, vector $x$ restricted to the
        entries of $E(G')$ can be written as a convex combination of
        2-edge-connected subgraphs of $G'$. 

\item Vector $x$ can be
        written as a convex combination of 2-edge-connected subgraphs
        of $G$.
\end{enumerate}
\end{restatable}

\begin{proof}
To show that $2. \Rightarrow 1.$, we observe that every
2-edge-connected subgraph of $G$ can be mapped to a 2-edge-connected
subgraph of $G'$ by considering only the subset of edges that belong
to $G'$.  Moreover, a convex combination corresponding to the vector
$x$ can be mapped to a convex combination for the vector $x'$, where
$x'$ contains only entries corresponding to edges in $G'$.
Now we show the other direction, namely $1. \Rightarrow 2.$

Suppose $G$ contains $t$ proper 3-edge cuts. We assume 1. and we prove
the statement 2. by induction on $t$.  If $t=0$, then $G$ does not
contain a proper 3-edge cut.  In this case, $G\in \mathcal{G}$ and we
are done.  If $t \geq 1$, find a minimal 3-edge cut $U$ of $G$. By
Observation \ref{obs:cornercut}, $G_U$ does not contain any proper
3-edge cuts and by Observation \ref{obs:inductivecriticalcut},
$G_{\bar{U}}$ contains at most $t-1$ proper 3-edge cuts.  Since $G_U$
belongs to $\mathcal{G}$, we can write $x$ restricted to $E(G_U)$ as a
convex combination of 2-edge-connected subgraphs of $G_U$.  By
induction, since $G_{\bar{U}}$ has at most $t-1$ proper 3-edge cuts,
we can write $x$ restricted to the entries of $E(G_{\bar{U}})$ as a
convex combination of 2-edge-connected subgraphs of $G_{\bar{U}}$.
Then we can apply Lemma \ref{gluing-in-polytime} to find a convex
combination of 2-edge-connected subgraphs for $G$.

Moreover, since $t$ is upper bounded by a polynomial in $|V|$, this
shows that the support of the convex combination for $G$ has
polynomial size.  Note that a minimal proper 3-edge cut of $G$ can be
found in polynomial time, since it is a minimum cut of $G$ and there
are only polynomially many minimum cuts of $G$.
\end{proof}

We will use the following observation in the proof of Theorem \ref{reductionToFourEC}.

\begin{observation}\label{cccontract}
Let $G=(V,E)$ be a 3-edge-connected graph and $\mathcal{C}$ be a cycle
cover of $G$ that covers the 3-edge cuts of $G$. Let $\emptyset\subset
S\subset V$ be a proper 3-edge cut of $G$ (i.e., $|\delta(S)| = 3$)
such that $|\delta(S) \cap \mathcal{C}| = 2$. Then the graph $G_S$ is
3-edge-connected and $\C$ restricted to $E(G_S)$ is a cycle cover of $G_S$.
\end{observation}

\reductionToEssentiallyFourEC*

\begin{proof}
Suppose we are given a 3-edge-connected graph $G=(V,E)$ and a cycle
cover $\C$ that covers all the 3-edge cuts of $G$.  Let $x = \chi^E -
\alpha \chi^{\C}$ for some $\alpha \in [0,1]$.  For each $G' \in
\mathcal{G}$, we apply Observation \ref{cccontract}, which yields
cycle cover $\C'$ (i.e., the cycle cover $\C$ restricted to the edges
of $G'$).  By Lemma \ref{gluing-essentially-4EC} we can conclude that
statements 1. and 2. are equivalent.
\end{proof}

\section{3-Edge-colorable cubic graphs: Proof of Theorem
  \ref{thm:3-edge-colorable}}\label{app:3edgecolorable}.

We now give the complete proof for the following theorem.
\edgeColorable*

Note that while deciding whether a graph is 3-edge-colorable is
NP-complete, there are special cases in which it is easy to find a
3-edge coloring, for example in the case of a bipartite, cubic graph.
Throughout this section, we assume that we are given a 3-edge coloring
of the input graph $G=(V,E)$.  We will consider three cycle covers of
$G$ where each cycle cover consists of all the edges on two of the
colors.  In to prove Theorem \ref{thm:3-edge-colorable} we prove the
following theorem.

\begin{theorem}\label{ccpoint-cubic-3edgecolorable}
	Let $G=(V,E)$ be a 3-edge-connected cubic, 3-edge-colorable
        graph and let $\mathcal{C}$ be a cycle cover of $G$ consisting
        of all edges of two colors.  Then the vector $\chi^E -
        \frac{1}{5}\chi^{\mathcal{C}}$ can be written as a convex
        combination of 2-edge-connected subgraphs of $G$.
\end{theorem}

\begin{proof}[Proof of Theorem \ref{thm:3-edge-colorable}]
Let $\sum_{i=1}^{3}\lambda_i \chi^{\mathcal{C}_i} = \frac{2}{3}
\chi^E$ be the convex combination of three cycle covers, each one
consisting of two edge colors.  By Theorem \ref{ccpoint-cubic}, for
each $i\in\{1,\ldots,3\}$ we can write $\chi^E -
\frac{1}{5}\chi^{\C_i}$ as convex combination of 2-edge-connected
subgraphs of $G$. Hence, $\sum_{i=1}^{3} \lambda_i (\chi^E -
\frac{1}{5}\chi^{\C_i}) = \chi^E - \frac{2}{15}\chi^E =
\frac{13}{15}\chi^E$, and we conclude that $\frac{13}{15} \chi^E$ can
be written as a convex combination of 2-edge-connected subgraphs of
$G$.
\end{proof}

Using the following observation, we can apply
Theorem \ref{reductionToFourEC}.  This allows us to focus on proving
Lemma \ref{1-1/5}, which implies Theorem \ref{ccpoint-cubic-3edgecolorable}.

\begin{observation}\label{contract3edgeColorable}
Let $G=(V,E)$ be a 3-edge-connected, 3-edge-colorable graph and
$\mathcal{C}$ be a cycle cover of $G$ that consists of all edges of
two colors.  Then (i) $\C$ covers the 3-edge cuts of $G$, (ii) $\C$
consists of only even-length cycles, (iii) $G_S$ is 3-edge colorable
(i.e., the 3-edge coloring of $G$ restricted to $E(G_S)$ is a 3-edge
coloring of $G_S$), and (iv) $\C$ restricted to $E(G_S)$ is also a
cycle cover of $G_S$ on two colors.
\end{observation}

\begin{lemma}\label{1-1/5}
Let $G=(V,E)$ be an essentially 4-edge-connected cubic graph and
$\mathcal{C}$ be a cycle cover of $G$ that consists of only
even-length cycles. Then the vector $\chi^E -
\frac{1}{5}\chi^{\mathcal{C}}$ can be written as a convex combination
of 2-edge-connected subgraphs of $G$.
\end{lemma}

\begin{proof}
Let $y= \chi^E - \frac{1}{2}\chi^{\C}$. Let $\sum_{i=1}^{k}\lambda_i
\chi^{K_i}$ be the convex combination of vector $y$ obtained via Lemma
\ref{rainbow2}.  
We now set $H_i = G/C_{K_i} = (U_i,F_i)$ and $T_i =
K_i/C_{K_i}$ (recall that $C_{K_i}$ is the unique cycle in the 1-tree
$K_i$).  Notice that since $\C$ contains only even cycles, there are
no leaf-matching links for $K_i$ in $G$, and consequently
$LM(H_i,T_i)$ is empty.
Thus, by Lemma \ref{lem:newlemma}, we can find a $(3\chi^{F_i\setminus
  T_i},5)$-coloring of $L_i = F_i\setminus T_i$ for $i \in
\{1,\ldots,k\}$. By Observation \ref{combination}, we have
$\frac{3}{5}\chi^{L_i}= \sum_{j=1}^{5} \frac{1}{5} \chi^{A^i_j}$ where
$A^i_j$ is a feasible augmentation for $T_i$ and therefore for $K_i$
(by Observation \ref{obs:connector2tree}).  In other words, $K_i +
A_j^i$ for $i \in \{1, \ldots, k\}$ and $j \in \{1, \ldots, 5\}$ is a
2-edge-connected subgraph of $G$.  Let $z^i_j = \chi^{K_i + A_j^i} $
be the characteristic vector of the corresponding 2-edge-connected
subgraph of $G$. Then, we have:
\begin{eqnarray*}
z^i = \sum_{j=1}^5\frac{1}{5} z^i_j = \sum_{j=1}^5 \frac{1}{5}
\chi^{K_i + A_j^i} = \chi^{K_i} + \sum_{j=1}^5 \frac{1}{5} \chi^{A_j^i} =
\chi^{K_i} + \frac{3}{5} \chi^{L_i}.
\end{eqnarray*}
Notice that $\sum_{i=1}^k \lambda_i \chi^{L_i} \leq
\frac{1}{2}\chi^{\cal{C}}$.  Next, we claim that the vector $z$ can be
written as a convex combination of 2-edge-connected subgraphs of $G$.
\begin{eqnarray*}
z  & = &  \sum_{k=1}^k \lambda_i z^i = \sum_{i=1}^k \lambda_i (\chi^{K_i} +
\frac{3}{5} \chi^{L_i}) = \sum_{i=1}^k \lambda_i \chi^{K_i} +
\frac{3}{5} \sum_{i=1}^k \lambda_i \chi^{L_i}\\  & \leq &
(\chi^E - \frac{1}{2} \chi^{\cal{C}}) + \frac{3}{10}\chi^{\cal{C}} = 
\chi^E - \frac{1}{5}\chi^{\cal{C}}.
\end{eqnarray*}
This concludes the proof of Lemma \ref{1-1/5}.
\end{proof}

\end{document}